\documentclass[11pt, draftcls, onecolumn]{IEEEtran}
\IEEEoverridecommandlockouts 
\usepackage[normalem]{ulem}
\usepackage{tabularx}
\usepackage{booktabs}
\usepackage{graphicx}
\usepackage{subfigure}
\usepackage{color}
\usepackage{epsfig}
\usepackage{amssymb}
\usepackage{amsmath}
\usepackage{amsthm}
\usepackage{latexsym}
\usepackage{setspace}
\usepackage{bbm}
\usepackage{float}
\usepackage{dsfont}
\usepackage{caption}
\usepackage{enumerate}
\usepackage[ruled,vlined]{algorithm2e}

\newcommand{\dP}{\mathrm{P}}
\newcommand{\dQ}{\mathrm{Q}}
\newcommand{\bPP}[1]{{\dP_{#1}}}
\newcommand{\bQQ}[1]{{\dQ_{#1}}}
\newcommand{\bPr}[1]{{\mathbb{P}}\left(#1\right)}

\newcommand{\bP}[2]{\mathrm{P}_{#1}\left({#2}\right)}
\newcommand{\bQ}[2]{\mathrm{Q}_{#1}\left({#2}\right)}

\newcommand{\bEE}[1]{{\mathbb{E}}\left[{#1}\right]}

\newcommand{\cA}{{\mathcal A}}
\newcommand{\cB}{{\mathcal B}}
\newcommand{\cC}{{\mathcal C}}
\newcommand{\cD}{{\mathcal D}}

\newcommand{\tF}{{\tilde F}}

\newcommand{\cK}{{\mathcal K}}

\newcommand{\cM}{{\mathcal M}}

\newcommand{\mN}{{\mathbb N}}

\newcommand{\cP}{{\mathcal P}}

\newcommand{\cT}{{\mathcal T}}

\newcommand{\cU}{{\mathcal U}}
\newcommand{\tU}{{\tilde{U}}}

\newcommand{\cV}{{\mathcal V}}
\newcommand{\tV}{{\tilde{V}}}

\newcommand{\tX}{\tilde{X}}
\newcommand{\cX}{{\mathcal X}}

\newcommand{\tY}{\tilde{Y}}
\newcommand{\cY}{{\mathcal Y}}
\newcommand{\cZ}{{\mathcal Z}}
\newcommand{\tZ}{{\tilde Z}}

\newcommand{\ep}{\varepsilon}

\newcommand{\indicator}{{\mathds{1}}}
\newcommand{\mc}{-\!\!\!\!\circ\!\!\!\!-}

\newtheorem{theorem}{Theorem}
\newtheorem{proposition}[theorem]{Proposition}
\newtheorem{corollary}[theorem]{Corollary}
\newtheorem*{corollary*}{Corollary}

\newtheorem*{lemma*}{Lemmas}

\theoremstyle{remark}
\newtheorem*{remark*}{Remark}
\newtheorem*{remarks*}{Remarks}
\theoremstyle{definition}

\newtheorem{remark}{Remark}

\allowdisplaybreaks

\newcommand{\textchange}[1]{#1}

\begin{document}

\setlength\textfloatsep{0pt}

\title{Strong Converse using Change of Measure Arguments\thanks{A
    preliminary version of this work was presented at the IEEE
    International Symposium on Information Theory, Vail, USA, 2018.}}  

\author{ \IEEEauthorblockN{Himanshu Tyagi$^\dag$} and
  \IEEEauthorblockN{Shun Watanabe$^\ddag$} }

\maketitle

{\renewcommand{\thefootnote}{}\footnotetext{
\noindent$^\dag$Department of Electrical Communication Engineering,
Indian Institute of Science, Bangalore 560012, India.  Email:
htyagi@ece.iisc.ernet.in.

\noindent$^\ddag$Department of Computer and Information Sciences,
Tokyo University of Agriculture and Technology, Tokyo 184-8588, Japan.
Email: shunwata@cc.tuat.ac.jp.  }}

\maketitle

\renewcommand{\thefootnote}{\arabic{footnote}}
\setcounter{footnote}{0}

\begin{abstract}
The strong converse for a coding theorem shows that the optimal
asymptotic rate possible with vanishing error cannot be improved by
allowing a fixed error. Building on a method introduced by Gu and
Effros for centralized coding problems, we develop a general and
simple recipe for proving strong converse that is applicable for 
distributed problems as well. Heuristically, our proof of strong
converse mimics the standard steps for proving a weak converse, except
that we apply those steps to a modified distribution obtained by
conditioning the original distribution on the event that no error
occurs. A key component of our recipe is the replacement of the hard
Markov constraints implied by the distributed nature of the problem
with a soft information cost using a variational formula introduced by
Oohama. We illustrate our method by providing a short proof of the
strong converse for the Wyner-Ziv problem and strong converse theorems
for interactive function computation, common randomness and secret key
agreement, and the wiretap channel; \textchange{the latter three
  strong converse problems were open prior to this work.} 
\end{abstract}

\section{Introduction}
A coding theorem in information theory characterizes the optimal rate
such that there exists a code of that rate for the problem
studied. Often, the first version of such theorems are proved assuming
a vanishing probability of error criterion. This criterion facilitates
a simple proof relying on chain rules and Fano's inequality. The
strong converse holds for a coding theorem if the optimal rate claimed
by the theorem cannot be improved even if a fixed error is
allowed. The first strong converse was shown for the point-to-point
channel coding theorem and source coding theorem by Wolfowitz
(see~\cite{Wol61}). A general method for proving strong converse for
coding theorems in multiterminal information theory was introduced
in~\cite{AhlGacKor76}. This method uses a strong converse
for the image-size characterization problem, which is in turn shown using the blowing-up lemma; see~\cite{CsiKor11} for a comprehensive treatment. The approach
based on blowing-up lemma entails, in essence, changing the code to a
list-code with a list-size of vanishing rate. \textchange{Related
  recent works that
  involve a change in the underlying code, too, but use
modern tools from functional inequalities and measure concentration
literature
include\footnote{For another use of the Gaussian-Poincar\'e
  inequality in information theory, see~\cite{PolVer14}.}  
\cite{FonTan17b}
 and\footnote{The approach in \cite{LiuHanVer17} was extended in
   \cite{Liu18} to derive a dispersion converse bound for 
 the Wyner-Ahlswede-K\"orner network.} \cite{LiuHanVer17}.} 

In this work, we present a simple method for proving strong converses
for multiterminal problems that uses very similar steps as the weak
converse proofs.  Our method consists of two steps, both building on techniques available in the literature. The first step is a {\em change of measure
  argument}\footnote{Our argument differs from the change of measure
  argument used to prove sphere-packing bounds
  (\textchange{cf.~\cite{CsiKor11, Har68}}).} due to Gu and Effros \cite{GuEff09,
  GuEff11}. The key idea is to evaluate the performance of a given
code not under the original product measure, but under another
modified measure which depends on the code and under which the code is
error-free. Thus, when the standard rate bounds are applied along with Fano's inequality, we get a bound involving information quantities for the tilted measure, but without the Fano correction term for the error.

In \cite{GuEff09, GuEff11}, Gu and Effros applied the change of measure argument
for proving strong converse for source coding problems where there exists a terminal that observes all the random variables involved; a particular example
is the Gray-Wyner (GW) problem \cite{GraWyn:74}. A difficulty in
extending this approach to other distributed source coding problems is
the Markov chain constraints among random variables implied by the
information structure of the communication. Specifically, these Markov
chain constraints might be violated when the measure is switched. This
technical difficulty was circumvented in \cite{Wat17} for the
Wyner-Ahlswede-K\"orner (WAK) problem~\cite{AhlKor75, Wyn75iii},
i.e., the problem of lossless source coding with coded side
information, by relating the WAK problem to an extreme case of the GW
problem.  In this paper, we develop a more direct and general recipe
for applying the change of measure argument to various distributed
coding problems.

The second step of our recipe is the replacement of the hard Markov
chain and functional constraints by soft information cost penalties
using variational formulae introduced by Oohama in a series of papers
including \cite{Ooh15, Ooh16}.  These variational formulae involve
optimization over a nonnegative Lagrange multiplier, with the optimum
corresponding to the form with Markov constraints. In fact, when the
change of measure step is applied some of the distributions that need to
preserved, such as the channel transition probabilities, may
change. These, too, can be accommodated by a KL-divergence cost
constraint. At a high level, we  replace all ``hard'' information
constraints by ``soft'' divergence costs and  
complete the proof of strong converse by establishing super- or sub-additivity
of the resulting penalized rate functions.

As an illustration of this approach, consider the lossless source
coding problem; even though this problem does not involve any Markov chain constraint, it illustrates the essential ideas involved in our approach. 
Suppose that an independent and identically distributed (i.i.d.)
source $Z^n$ is compressed to $\varphi(Z^n)$ such that there exists a
function $\psi$ satisfying $\bPr{\psi(\varphi(Z^n)) = Z^n}\geq
1-\ep$. Let $\cC$ denote the set $\{z^n: \psi(\varphi(z^n)) =z^n\}$ of
sequences where no error occurs. The strong converse for the lossless
source coding theorem will be obtained upon showing that the rate of the code is bounded below by
entropy $H(Z)$ asymptotically, irrespective of the value of $0 < \ep < 1$. To show this, we change the probability measure to
$\bPP{\tilde{Z}^n}$ defined by\footnote{In \cite{GuEff09, GuEff11}, the new distribution had a more complicated form.}
\begin{align} \label{eq:changed-measure-introduction}
\bP{\tilde{Z}^n}{z^n} = \bPr{ Z^n = z^n | Z^n \in \cC}.
\end{align}
This measure is not too far from the original measure under
KL-divergence. Indeed,\footnote{\textchange{This simple, but important, observation was used in \cite{Mar86} to provide a simple proof of the blowing-up
    lemma.}}
\[
D(\bPP{\tilde{Z}^n} \|\bPP{Z^n}) \le \log \frac 1 {1-\varepsilon}.
\]
 On the other hand, under $\bPP{\tilde{Z}^n}$, the error probability
 of the code $(\varphi,\psi)$ is exactly zero. Thus, by mimicking the
 standard weak converse arguments, we have
\begin{align*}
\log |\cC| \ge H(\tilde{Z}^n).
\end{align*}
The next step is to single-letterize $H(\tilde{Z}^n)$, which now does
not correspond to a product measure and may not be super-additive on
its own. We circumvent this difficulty by adding a divergence cost to
get
\begin{align}
\frac{1}{n} \log |\cC| &\ge \frac{1}{n} H(\tilde{Z}^n) +
\frac{\alpha}{n} \bigg[ D(\bPP{\tilde{Z}^n} \|\bPP{Z^n}) - \log
  (1/(1-\varepsilon)) \bigg] \nonumber \\
&\ge \min_{\bPP{\tilde{Z}}} \big[ H(\tilde{Z}) + \alpha
  D(\bPP{\tilde{Z}} \| \bPP{Z}) \big] - \frac{\alpha \log
  (1/(1-\varepsilon)) }{n}, \nonumber
\end{align}
for any $\alpha > 0$. The second inequality uses a simple
super-additivity property that we show in
Proposition~\ref{proposition:almost-subadditivity-entropy} for
conditional entropy. The proof of strong converse is completed by
using the following variational formula for entropy:
\begin{align*}
H(Z) = \sup_{\alpha > 0} \min_{\bPP{\tilde{Z}}} \big[ H(\tilde{Z}) +
  \alpha D(\bPP{\tilde{Z}} \| \bPP{Z}) \big].
\end{align*}

Using our recipe, we can obtain simple proof for some known strong
converse results and can, in fact, obtain several new strong converse
results, including for problems involving interactive
communication. The first result we present is the lossy source coding
with side information problem, also known as the Wyner-Ziv (WZ)
problem \cite{WynZiv76}.  The strong converse for the WZ problem was
proved only recently in \cite{Ooh16}. We use our general recipe for
proving a strong converse to give a more compact proof for the WZ
strong converse which, we believe, is more accessible than the
original proof of \cite{Ooh16}.\footnote{The proof in \cite{Ooh16}
  provides a stronger result in form of an explicit lower bound on the
  exponent of the probability of correctness.} The second problem we
consider is the interactive function computation problem
(cf.~\cite{OrlRoc01, MaIsh11, BraRao11}).  Prior to our work, a strong
converse for this well-studied problem was unavailable. A technical
difficulty in showing such a result arises from the multiple auxiliary
random variables and Markov chain constraints that appear in the
optimal sum-rate.  \textchange{The strong converse for the
  interactive function computation problem has attracted attention in
  the theoretical computer science community as well, in the context of direct product theorems in communication complexity. A version of the strong converse
  result was shown in \cite{BraWei15} in a slightly different setting,
  but the basic strong converse itself has been open. Furthermore, the
  information odometer approach used in~\cite{BraWei15} is technically
  much more involved than our simple change of measure argument.}

In addition to the two source coding problems mentioned above, we also apply our recipe for problems of generating common randomness and secret key
with interactive communication \cite{AhlCsi98, Tya13}. The strong converse for these problems with interactive communication were unavailable prior to our work; see~\cite{HayTW14} and the extended version of~\cite{LiuHanVer17} for partial results. Since these problems involve a total variation distance constraint, we need some additional tricks for changing measure. In particular, we seek a replacement for the correctly decoded set of sequences $\cC$. We illustrate the essential idea using
a simple random number generation problem, which is also known as the intrinsic randomness problem (cf.~\cite{Han03}). Suppose that an i.i.d. source $Z^n$ is converted to
$K = \varphi(Z^n)$ such that the total variation distance criterion $d(\bPP{K}, \bPP{\mathtt{unif}})\le \delta$
is satisfied, where $\bPP{\mathtt{unif}}$ is the uniform distribution of the range $\cK$ of $K$. Consider the set
\begin{align} \label{eq:high-entropy-density-set}
{\cal C} = \bigg\{ z^n : \log \frac{1}{\bP{K}{\varphi(z^n)}} \ge \log |\cK | - \log (2/(1-\delta)) \bigg\}
\end{align}
comprising elements $z^n$ mapped to high entropy density realizations of $K$.
It can be seen from our analysis in Section~\ref{sec:CR-SK} that
\[
\bPr{Z^n \in {\cal C}} \ge \frac{1-\delta}{2}.
\]
Thus, by changing the measure to $\bPP{\tilde{Z}^n}$ given in \eqref{eq:changed-measure-introduction}
but using the set ${\cal C}$ of \eqref{eq:high-entropy-density-set}, we have $D(\bPP{\tilde{Z}^n} \|\bPP{Z^n}) \le \log (2/(1-\delta))$.
Furthermore, for this changed measure, the random variable $\tilde{K} = \varphi(\tilde{Z}^n)$ has the min-entropy at least
$\log |\cK| - 2\log(2/(1-\delta))$, which implies 
\begin{align*}
\log |\cK| &\le H_{\min}(\tilde{K}) + 2 \log (2/(1-\delta)) \\
&\le H(\tilde{K}) + 2 \log (2/(1-\delta)) \\
&\le H(\tilde{Z}^n) + 2 \log (2/(1-\delta)).
\end{align*}
The entropy term on the right-side can be bounded using the sub-additivity of  entropy. However, the resulting single-letterized measure may deviate from the original $\bPP{Z}$, which needs to be retained. To that end, we add a 
divergence cost to get 
\begin{align*}
\frac{1}{n} \log |\cK| &\le \frac{1}{n} H(\tilde{Z}^n) - \frac{\alpha}{n}\bigg[ D\left(\bPP{\tilde{Z}^n} \| \bPP{Z^n}\right) - \log (2/(1-\delta)) \bigg] + \frac{2 \log (2/(1-\delta))}{n} \\
&\le \max_{\bPP{\tilde{Z}}} \big[ H(\tilde{Z}) - \alpha D(\bPP{\tilde{Z}} \| \bPP{Z}) \big] + \frac{(\alpha+2)\log(2/(1-\delta))}{n},
\end{align*}
for any $\alpha > 0$. The strong converse for the random number generation problem follows from the variational formula
\begin{align*}
H(Z) = \inf_{\alpha > 0} \max_{\bPP{\tilde{Z}}}\big[ H(\tilde{Z}) - \alpha D(\bPP{\tilde{Z}} \| \bPP{Z}) \big].
\end{align*}

The final setting we consider is the wiretap channel \cite{Wyn75ii, CsiKor78}.
The strong converse theorem for degraded wiretap channel was proved in \cite{HayTyaWat14iii}
(see \cite{TanBlo15} for a partial strong converse). However, its extension to general wiretap channel has remained open.\footnote{The argument in \cite{WeiUlu16} has a technical flaw, 
and we are unable to verify the technically involved proof-sketch in the conference paper~\cite{GraWon17}; \textchange{a full-version of \cite{GraWon17} has not been published so far.}} By using our general recipe, we provide a proof for the strong converse
theorem for the general wiretap channel. Compared to other problems mentioned above, this problem is more involved,
and requires a few more tricks including the expurgation of messages to
replace average guarantees with worst-case guarantees and the
construction of the changed measure using a set with bounded
log-likelihood ratio of wiretappers observation probability and its
probability given the message. Nevertheless, given the technical
difficulties in prior attempts, this is a relatively simple proof.  

Overall, our main message in this work is that strong converses can be
proven using similar techniques as those used for proving weak
converses, applied after an appropriate change of measure. However, we
need to work with new variational forms of capacity formulae where the
hard information constraints are replaced with soft KL-divergence
costs. 

A conceptually related approach for proving strong converse was recently
proposed by Kosut and Kliewer in \cite{KosKli17}. In their approach, 
the strong converse for a given network is reduced first to the weak
converse by adding an extra edge of vanishing rate to the network,
which allows negligible cooperation among users. Then, the strong
converse will follow if the so-called edge removal property holds,
namely the capacity is not changed when the extra edge is
removed. Since Markov chain constraints in multiterminal problems stem
from distributed nature of the problems,  
the replacement of those Markov chain constraints with soft
KL-divergence costs in our recipe is, at high-level, similar to adding
a ``soft edge'' to increase cooperation among the terminals. However,
the soft divergence cost seems to be a more versatile tool; in
particular, it allows us to handle even interactive communication. 

Another related recent work is that of Jose and Kulkarni
\cite{JosKul17, JosKul18}. Their approach considers the performance of
the optimal code for a coding problem and poses it as an optimization
problem, which is further bounded by the value of a linear program
obtained by relaxing some constraints.  Even though this approach
provides tight converse bounds implying strong converse for some
problems, applicability of this approach to problems involving
auxiliary random variables is unclear.

\textchange{In a slightly different direction, Fong
  and Tan \cite{FonTan19} proved strong converse theorems for
  multi-message networks with tight cut-set bound, such as the degraded
  relay channel and relay channel with orthogonal components, for
  both discrete and Gaussian channels. These results are inspired by the result for the reliability 
  function of a DMC with feedback above capacity \cite{CsiKor82}, and
  are different in nature than our setting. In particular, these results do not include multiple auxiliary random variables,
  which is a significant difficulty we overcome in this paper.}
  
The remainder of the paper is organized as follows. We begin by
reviewing a few simple results in the next section, which will be used 
throughout the paper. The strong converse for the WZ problem is given
in 
Section~\ref{s:WZ} and for the function computation problem in
Section~\ref{s:FC}. The next two sections contain problems involving
total variation constraints, with the common randomness generation and  
secret key agreement in Section~\ref{sec:CR-SK}, and the wiretap
channel problem in Section~\ref{sec:wiretap}.  We conclude with
discussions on exponential strong converse and extensions in the final
section.

\paragraph*{Notation} Throughout the paper, we restrict to
discrete random variables taking finitely many values and denote the
random variable with a capital letter, for instance $X$, its range-set
with the corresponding calligraphic, e.g.~$\cX$, and each realization
with a small letter, e.g.~$x$. For information measures, we follow the
standard notations in \cite{CsiKor11}: The entropy, the KL divergence,
and the mutual information are denoted by $H(X)$, $D(\bPP{} \|
\bQQ{})$, and $I(X \wedge Y)$, respectively.  The total
variation distance between two distributions $\bPP{}$ and $\bQQ{}$ is
denoted by $d(\bPP{}, \bQQ{}) := \frac{1}{2} \sum_x |\bP{}{x}-
\bQ{}{x}|$.  For a sequence $X^n = (X_1,\cdots,X_n)$ of random
variables, we denote $X_j^- = (X_1,\ldots,X_{j-1})$ and $X_j^+ =
(X_{j+1},\ldots,X_n)$, where $X_1^-$ and $X_n^+$ are regarded as the
empty string.  The indicator function is denoted by
$\indicator[\cdot]$.  Other notations will be introduced when
necessary, but are standard notations used in the multiterminal
information theory literature.

\section{Technical Tools}
We begin by assembling the simple tools that we will use repeatedly in
our proofs. The first is perhaps a new observation; the other two are
standard.

Typically, we use additivity of (conditional) entropy for independent
random variables for proving converse bounds. However, in our proofs,
once we change the measure, the resulting random variables need not be
independent. Nevertheless, the following simple result fills the gap
and shows that if we add a divergence cost for change of measure, the sum
is super-additive.
\begin{proposition} \label{proposition:almost-subadditivity-entropy}
For i.i.d. $\bPP{X^n Y^n}$ with common distribution $\bPP{XY}$ and any
$\bPP{\tilde{X}^n\tilde{Y}^n}$, we have
\begin{align*}
H(\tilde{X}^n|\tilde{Y}^n) + D(\bPP{\tilde{X}^n \tilde{Y}^n} \|
\bPP{X^n Y^n}) \ge n \big[ H(\tilde{X}_J | \tilde{Y}_J) +
  D(\bPP{\tilde{X}_J \tilde{Y}_J} \| \bPP{XY}) \big],
\end{align*}
where $J\sim {\tt unif}(\{1,...,n\})$ is the time-sharing random
variable and is assumed to be independent of all the other random
variables involved.
\end{proposition}
\begin{proof}
The left-side can be expressed as
\begin{align*}
H(\tilde{X}^n|\tilde{Y}^n) + D(\bPP{\tilde{X}^n | \tilde{Y}^n} \|
\bPP{X^n |Y^n} | \bPP{\tilde{Y}^n}) + D(\bPP{\tilde{Y}^n} \|
\bPP{Y^n}).
\end{align*}
The sum of the first two terms satisfy
\begin{align*}
H(\tilde{X}^n|\tilde{Y}^n) + D(\bPP{\tilde{X}^n | \tilde{Y}^n} \|
\bPP{X^n |Y^n} | \bPP{\tilde{Y}^n}) &= \sum_{x^n, y^n} \bP{\tilde{X}^n
  \tilde{Y}^n}{x^n,y^n} \log \frac{1}{\bP{X^n | Y^n}{x^n|y^n}} \\ &=
\sum_{j=1}^n \sum_{x,y} \bP{\tilde{X}_j \tilde{Y}_j}{x,y} \log
\frac{1}{\bP{X|Y}{x|y}} \\ &= n \sum_{x,y} \bP{\tilde{X}_J
  \tilde{Y}_J}{x,y} \log \frac{1}{\bP{X|Y}{x|y}} \\ &= n H(\tilde{X}_J
| \tilde{Y}_J) + n D(\bPP{\tilde{X}_J | \tilde{Y}_J} \| \bPP{X|Y} |
\bPP{\tilde{Y}_J}),
\end{align*}
and the third satisfies
\begin{align*}
D(\bPP{\tilde{Y}^n} \| \bPP{Y^n}) &= \sum_{j=1}^n D(\bPP{\tilde{Y}_j
  |\tilde{Y}_j^-} \| \bPP{Y} | \bPP{\tilde{Y}_j^-}) \\ &\ge
\sum_{j=1}^n D(\bPP{\tilde{Y}_j} \| \bPP{Y}) \\ &\ge n
D(\bPP{\tilde{Y}_J} \| \bPP{Y}),
\end{align*}
which completes the proof.
\end{proof}
The next tool we present is essential for handling the distributed
settings we consider. It allows us to replace the ``hard'' Markov
chain and function constraints in our bounds with ``soft'' costs using
a variational formula introduced by Oohama (cf.~\cite{Ooh16}) in this
context.  This is important since these hard constraints may not hold
once we change the measure.
We describe this approach in an abstract form below; proofs for
specific variants needed for our results are similar and have been
relegated to the appendix.

Let $G(\bPP{Z_1Z_2})$ be a bounded continuous function of
$\bPP{Z_1Z_2}$. Define\footnote{We abbreviate
  $G(\bP{Z_1Z_2|U}{\cdot|U})$ as $G(\bPP{Z_1Z_2|U})$.}
\[
\overline{G}(\bPP{Z_1Z_2}) =
\inf_{\bPP{U|Z_1Z_2}: \bPP{U|Z_1Z_2}=\bPP{U|Z_1}}\bEE{G(\bPP{Z_1Z_2|U})}.
\]
Note that by the support lemma~\cite{CsiKor11}, it suffices to
restrict the infimum to $U$ with $|\cU|\leq |\cZ_1|$, and thereby the
$\inf$ can be replaced by $\min$ using compactness of the finite
dimensional probability simplex. The next result we present is a
variational formula for $\overline{G}(\bPP{Z_1Z_2})$ that allows us to
replace the minimization over $U$ satisfying the Markov chain
condition $U \mc Z_1 \mc Z_2$ to that over all $\bPP{U|Z_1Z_2}$.
\begin{proposition}\label{p:remove_Markov}
Let $G(\bPP{Z_1Z_2})$ be a bounded continuous function over the
probability simplex $\cP(\cZ_1\times \cZ_2)$. Then, the function
$\overline{G}(\bPP{Z_1Z_2})$
satisfies
\begin{align}
\overline{G}(\bPP{Z_1Z_2}) = \sup_{\alpha > 0} \min_{\bPP{U|Z_1Z_2}}
\bigg[ \bEE{G(\bPP{Z_1Z_2|U})} + \alpha I(U\wedge Z_2|Z_1) \bigg].
\label{eq:abstract-variational-formula}
\end{align}
\end{proposition}
\begin{proof}
The left-side is greater than or equal to the right-side since, for
every $\alpha>0$, the left-side is obtained by restricting the inner
minimization on the right to the distribution satisfying $U\mc Z_1\mc
Z_2$. To prove the other direction, first note that $I(U \wedge
Z_1|Z_2)$ can be written as $D(\bPP{UZ_1Z_2} \| \bPP{U|Z_1} \bPP{Z_1}
\bPP{Z_2|Z_1})$. Given $\alpha > 0$, let
$\mathrm{P}^\alpha_{U|Z_1Z_2}$ attain the inner minimum in
\eqref{eq:abstract-variational-formula} for $\alpha$. Since the
function $G(\cdot)$ is bounded, say it lies in an interval $[a,b]$,
the same holds for the function $\overline{G}(\cdot)$. Therefore, it
must hold that $D(\mathrm{P}^\alpha_{UZ_1Z_2} \|
\mathrm{P}^\alpha_{U|Z_1} \bPP{Z_1Z_2}) \le (b-a)/\alpha$. Let
$\tilde{\mathrm{P}}_{UZ_1Z_2}=
\mathrm{P}^\alpha_{U|Z_1}\bPP{Z_1Z_2}$. Since $G(\cdot)$ is
continuous\footnote{We are assuming $G(\cdot)$ is continuous with
  respect to the total variation distance. Then, it is also continuous
  with respect to the KL divergence using Pinsker's inequality.}  on
a compact domain, it is also uniformly continuous. Therefore, there
exists a function $\Delta(t)$ satisfying $\Delta(t)\to 0$ as $t\to0$
such that
\begin{align*}
\bEE{G(\mathrm{P}^\alpha_{Z_1Z_2|U})} &\ge
\bEE{G(\tilde{\mathrm{P}}_{Z_1Z_2|U})} - \Delta((b-a)/\alpha) \\ &\ge
\overline{G}(\bPP{Z_1Z_2}) - \Delta((b-a)/\alpha).
\end{align*}
Thus, we obtain the required inequality by taking $\alpha \to \infty$,
which completes the proof.
\end{proof}
The variational form above can be used to handle even multiple Markov
relations by adding a similar cost for each constraint. Furthermore,
we can even handle functional constraints such as $H(Z_1|U,Z_2) =0$ by
adding an additional cost $\alpha H(Z_1|U, Z_2)$. These extensions of
Proposition~\ref{p:remove_Markov} will be used in our proofs.

Additionally, we also need a cost to account for the deviation from
the underlying fixed source and channel distributions that occur when
we apply our change of measure arguments. The following alternative
variational formula for $\overline{G}(\bPP{Z_1Z_2})$ will be handy:
\begin{proposition} \label{p:remove_Markov2}
Let $G(\bPP{Z_1Z_2})$ be a bounded continuous function over the
probability simplex $\cP(\cZ_1\times \cZ_2)$. Then, we have
\begin{align*}
\overline{G}(\bPP{Z_1Z_2}) = \sup_{\alpha > 0}
\min_{\bPP{\tilde{U}\tilde{Z}_1\tilde{Z}_2}} \bigg[
  \bEE{G(\bPP{\tilde{Z}_1\tilde{Z}_2|\tilde{U}})} + \alpha
  \big(D(\bPP{\tilde{Z}_1\tilde{Z}_2} \| \bPP{Z_1Z_2}) + I(U\wedge
  Z_2|Z_1) \big) \bigg].
\end{align*}
\end{proposition}
The proof is similar to that of Proposition~\ref{p:remove_Markov};
instead of proving this meta-result, we will prove our specific
variational formulae in the appendix.

The final result we recall is a standard tool for single-letterization
from~\cite[pg. 314]{CsiKor11}-- its power lies in its validity for
arbitrary distributions. For random variables $X^n, Y^n, U$ with an
arbitrary joint distribution $\bPP{X^n Y^n U}$, it holds that
\begin{align}
H(X^n|U)- H(Y^n|U)=\sum_{i=1}^n H(X_i| X_i^-,Y_i^+, U) - H(Y_i|
X_i^-,Y_i^+, U).
\label{e:CKM_inequality}
\end{align}

\section{Lossy source coding with side-information}\label{s:WZ}
In the lossy source coding problem with side-information, the goal is
to compress a source sequence to enable its recovery within a
prespecified distortion at a receiver with side-information.
Formally, for a given source $\bPP{XY}$ on a finite alphabet ${\cal
  X}\times {\cal Y}$, a lossy source code with side-information
consists of an encoder $\varphi: {\cal X}^n \to {\cal M}$ and a
decoder $\psi:{\cal M} \times {\cal Y}^n \to {\cal Z}^n$, where ${\cal
  Z}$ is the reproduction alphabet. Consider a distortion measure
$d:{\cal X}\times {\cal Z} \to [0, D_{\max}]$ and its $n$-fold
extension $d(x^n,z^n) = \sum_{i=1}^n d(x_i,z_i)$. A rate-distortion
pair $(R,D)$ is $\varepsilon$-achievable if, for every sufficiently
large $n$, there exists a code $(\varphi,\psi)$ such that
\begin{align} \label{eq:WZ-excess-probability}
\bPr{ d(X^n, \psi(\varphi(X^n),Y^n)) > nD } \le \varepsilon
\end{align}
and
\begin{align} \label{eq:WZ-rate}
\frac{1}{n} \log |{\cal M}| \le R.
\end{align}
Let ${\cal R}_{\mathtt{WZ}}(\varepsilon|\bPP{XY})$ be the closure of
the set of all $\varepsilon$-achievable rate-distortion pairs. Define
\begin{align*}
{\cal R}_{\mathtt{WZ}}(\bPP{XY}) := \bigcap_{0 < \varepsilon < 1}
{\cal R}_{\mathtt{WZ}}(\varepsilon|\bPP{XY}).
\end{align*}
The following characterization\footnote{In fact, we can restrict $Z$
  to be a function of $(U,Y)$.} of ${\cal R}_{\mathtt{WZ}}(\bPP{XY})$
was given in \cite{WynZiv76}:
\begin{align*}
{\cal R}_{\mathtt{WZ}}(\bPP{XY}) &= \{ (R,D): \exists\, (U,Z) \mbox{
  s.t. }  |{\cal U}|\le |{\cal X}|+1, \\ &~~~~~~U \mc X \mc Y, Z \mc
(U,Y) \mc X, \\ &~~~~~~ R \ge I(U\wedge X|Y), \mathbb{E}[d(X,Z)]\le
D\}.
\end{align*}
The set ${\cal R}_{\mathtt{WZ}}(\bPP{XY})$ is closed and convex and
can be expressed alternatively using tangent lines as follows:
\begin{align*}
{\cal R}_{\mathtt{WZ}}(\bPP{XY}) &= \bigcap_{\mu \ge 0} \{ (R,D): R +
\mu D \ge R_{\mathtt{WZ}}^\mu(\bPP{XY}) \},
\end{align*}
where
\begin{align*}
R_{\mathtt{WZ}}^\mu(\bPP{XY}) &:= \min\big\{ I(U \wedge X|Y) + \mu
\mathbb{E}[d(X,Z)] : \\ &~~~~~~~~~~~~ \exists\, (U,Z) \mbox{ s.t. }
|{\cal U}|\le |{\cal X}|, U \mc X \mc Y, Z \mc (U,Y) \mc X \big\}.
\end{align*}
The optimal rate region above involves Markov relations, which will
become intractable once we change the measure. Furthermore, once we
change the measure and obtain a single-letter bound, the source
distribution may deviate from $\bPP{XY}$. To circumvent these
difficulties, we switch to the following variational form of
$R_{\mathtt{WZ}}^\mu(\bPP{XY})$, which will be proved in Appendix
\ref{app:WZ}:
\begin{align}
R_{\mathtt{WZ}}^\mu(\bPP{XY}) = \sup_{\alpha > 0}
R_{\mathtt{WZ}}^{\mu,\alpha}(\bPP{XY}), \label{eq:WZ-variational-form}
\end{align} 
where
\begin{align}
R_{\mathtt{WZ}}^{\mu,\alpha}(\bPP{XY}) &:=
\min_{\bPP{\tilde{U}\tilde{X}\tilde{Y}\tilde{Z}}} \big[ I(\tilde{U}
  \wedge \tilde{X}|\tilde{Y}) + \mu \mathbb{E}[d(\tilde{X},\tilde{Z})]
  + \alpha D(\bPP{\tilde{U}\tilde{X}\tilde{Y}\tilde{Z}} \|
  \bQQ{\tilde{U}XY\tilde{Z}}) + D(\bPP{\tilde{X}\tilde{Y}}\|\bPP{XY})
  \big] \label{eq:WZ-divergence-form} \\ &=
\min_{\bPP{\tilde{U}\tilde{X}\tilde{Y}\tilde{Z}}} \big[ I(\tilde{U}
  \wedge \tilde{X}|\tilde{Y}) + \mu \mathbb{E}[d(\tilde{X},\tilde{Z})]
  \nonumber \\ &~~~+ \big( (\alpha+1) D(\bPP{\tilde{X}\tilde{Y}} \|
  \bPP{XY}) + \alpha I(\tilde{U} \wedge \tilde{Y}|\tilde{X}) + \alpha
  I(\tilde{Z} \wedge \tilde{X}|\tilde{U},\tilde{Y}) \big)
  \big] \label{eq:WZ-mutual-information-form}
\end{align}
and $\bQQ{\tilde{U}XY\tilde{Z}} = \bPP{\tilde{Z}|\tilde{U}\tilde{Y}}
\bPP{\tilde{U}|\tilde{X}} \bPP{XY}$ is the distribution induced from
each $\bPP{\tilde{U}\tilde{X}\tilde{Y}\tilde{Z}}$. \textchange{Note that this
$\bQQ{\tilde{U}XY\tilde{Z}}$ respects the information structure of the
coding problem; we will use this convention in our usage of notation $\bQQ{}$ throughout.} 
By the support
lemma~\cite{CsiKor11}, the range $\cU$ of $\tU$ can be restricted to
$|\cU|\le |\cX||\cY||\cZ|$.
\begin{remark}\label{r:two_form}
In effect, we have replaced the ``hard constraints'' imposed by the
requirements of preserving the input source distribution and the
Markov relations between the communication sent, the source and the
reconstructed estimate with ``soft'' divergence penalties which are
amenable to single-letterization using standard chain rules.  The
factor $(\alpha+1)$ instead of $\alpha$ is only to enable a technical
manipulation in the proof of Theorem \ref{theorem:WZ} below. However,
semantically, the bound can be understood by just considering an extra
$ \alpha D(\bPP{\tilde{U}\tilde{X}\tilde{Y}\tilde{Z}} \|
\bQQ{\tilde{U}XY\tilde{Z}})$ cost which captures all the
aforementioned constraints. In fact, a factor in the form of any
function $f(\alpha)$ of $\alpha$ that blows-up to infinity as $\alpha$
tends to infinity will work, since we take $\alpha \to \infty$ at the
end.  In the definition of $R_{\mathtt{WZ}}^{\mu,\alpha}(\bPP{XY}) $,
the divergence form \eqref{eq:WZ-divergence-form} is heuristically
appealing and affords a simple proof of the variational formula
\eqref{eq:WZ-variational-form} (see Appendix \ref{app:WZ}); on the
other hand, the mutual information form
\eqref{eq:WZ-mutual-information-form} is amenable to
single-letterization in the proof of Theorem \ref{theorem:WZ} below.
\end{remark}

We are now in a position to prove the strong converse. The main step
is to show the following result, which is obtained simply by using the
super-additivity of the lower bound obtained after change of measure.
\begin{theorem} \label{theorem:WZ}
For every $n \in \mathbb{N}$, $\mu \ge 0$, and $\alpha > 0$, we have
\begin{align*}
R_{\mathtt{WZ}}^{\mu,\alpha}(\mathrm{P}_{XY}^n) \ge n
R_{\mathtt{WZ}}^{\mu,\alpha}(\bPP{XY}).
\end{align*}
\end{theorem}
As a corollary, we obtain the strong converse for the lossy source
coding with side-information problem, which was shown in \cite{Ooh16}
using a different, more complicated method.
\begin{corollary} \label{corollary:WZ}
For every $0 < \varepsilon < 1$, we have $ {\cal
  R}_{\mathtt{WZ}}(\varepsilon|\bPP{XY})= {\cal
  R}_{\mathtt{WZ}}(\bPP{XY})$.
\end{corollary}

\paragraph*{Proof of Corollary \ref{corollary:WZ}}
For a given code $(\varphi,\psi)$ satisfying
\eqref{eq:WZ-excess-probability} and \eqref{eq:WZ-rate}, define
\[
{\cal D} = \{ (x^n,y^n) : d(x^n, \psi(\varphi(x^n),y^n)) \le nD \}.
\]
Further, let $\bPP{\tilde{X}^n\tilde{Y}^n}$ be defined by
\begin{align*}
\bP{\tilde{X}^n \tilde{Y}^n}{x^n,y^n} :=
\frac{\mathrm{P}_{XY}^n(x^n,y^n) \indicator[(x^n,y^n) \in {\cal
      D}]}{\mathrm{P}_{XY}^n({\cal D})}.
\end{align*}
Then, the excess distortion probability of the same code
$(\varphi,\psi)$ for the source $(\tilde{X}^n,\tilde{Y}^n)$ is exactly
$0$, which implies $\tilde{Z}^n =
\psi(\varphi(\tilde{X}^n),\tilde{Y}^n)$ satisfies
$\mathbb{E}[d(\tilde{X}^n,\tilde{Z}^n)] \le nD$. Thus, by mimicking
the standard weak converse proof, we have
\begin{align*}
  n (R + \mu D)&\ge I(\tilde{S} \wedge \tilde{X}^n|\tilde{Y}^n) + \mu
  \mathbb{E}[d(\tilde{X}^n,\tilde{Z}^n)],
\end{align*}
where $\tilde{S} = \varphi(\tilde{X}^n)$.  Also,
\begin{align*}
D(\bPP{\tilde{X}^n \tilde{Y}^n} \| \mathrm{P}_{XY}^n) = \log
\frac{1}{\mathrm{P}_{XY}^n({\cal D})} \nonumber \le \log
\frac{1}{1-\varepsilon}.
\end{align*}
Thus, by noting that costs $I(\tilde{S} \wedge
\tilde{Y}^n|\tilde{X}^n)$ and $I(\tilde{Z}^n \wedge
\tilde{X}^n|\tilde{S},\tilde{Y}^n)$ are both $0$, we have
\begin{align*}
n (R + \mu D)&\ge I(\tilde{S} \wedge \tilde{X}^n|\tilde{Y}^n) + \mu
\mathbb{E}[d(\tilde{X}^n,\tilde{Z}^n)] + \big(
(\alpha+1)D(\bPP{\tilde{X}^n \tilde{Y}^n} \| \mathrm{P}_{XY}^n)
\\ &~~~ + \alpha I(\tilde{S} \wedge \tilde{Y}^n|\tilde{X}^n) + \alpha
I(\tilde{Z}^n \wedge \tilde{X}^n|\tilde{S},\tilde{Y}^n) \big) -
(\alpha+1)\log \frac{1}{1-\varepsilon} \\ &\ge
R_{\mathtt{WZ}}^{\mu,\alpha}(\mathrm{P}_{XY}^n) - (\alpha+1)\log
\frac{1}{1-\varepsilon}.
\end{align*}
Therefore, by Theorem~\ref{theorem:WZ}, we have
\begin{align}
R + \mu D \ge R_{\mathtt{WZ}}^{\mu,\alpha}(\bPP{XY}) -
\frac{(\alpha+1)}{n} \log \frac 1
     {1-\varepsilon} \label{eq:final-bound-WZ}
\end{align}
for every $\mu\geq 0$ and $\alpha>0$, whereby the corollary follows
from~\eqref{eq:WZ-variational-form}.  \qed

\paragraph*{Proof of Theorem \ref{theorem:WZ}}
By setting
\begin{align*}
G_1(\bPP{\tilde{X}^n\tilde{Y}^n}) &:= H(\tilde{X}^n|\tilde{Y}^n) +
\alpha H(\tilde{Y}^n|\tilde{X}^n) + (\alpha+1)
D(\bPP{\tilde{X}^n\tilde{Y}^n} \| \mathrm{P}_{XY}^n),
\\ G_2(\bPP{\tilde{U}\tilde{X}^n\tilde{Y}^n\tilde{Z}^n}) &:= -
H(\tilde{X}^n|\tilde{U},\tilde{Y}^n) + \mu
\mathbb{E}[d(\tilde{X}^n,\tilde{Z}^n)] + \alpha\big( -
H(\tilde{Y}^n|\tilde{U},\tilde{X}^n) + I(\tilde{Z}^n \wedge
\tilde{X}^n|\tilde{U},\tilde{Y}^n) \big),
\end{align*}
for given $\bPP{\tilde{U}\tilde{X}^n\tilde{Y}^n\tilde{Z}^n}$, we can
write
\begin{align*}
R_{\mathtt{WZ}}^{\mu,\alpha}(\mathrm{P}_{XY}^n) =
\min_{\bPP{\tilde{U}\tilde{X}^n\tilde{Y}^n\tilde{Z}^n}}\big[
  G_1(\bPP{\tilde{X}^n\tilde{Y}^n}) +
  G_2(\bPP{\tilde{U}\tilde{X}^n\tilde{Y}^n\tilde{Z}^n}) \big].
\end{align*}
Fix arbitrary $\bPP{\tilde{U}\tilde{X}^n\tilde{Y}^n\tilde{Z}^n}$.  By
Proposition \ref{proposition:almost-subadditivity-entropy},
$G_1(\bPP{\tilde{X}^n\tilde{Y}^n})$ can be lower bounded
as\footnote{By a slight abuse of notation, $G_1(\bPP{\tilde{X}_J
    \tilde{Y}_J})$ is defined by replacing
  $\bPP{\tilde{X}^n\tilde{Y}^n}$ and $\mathrm{P}_{XY}^n$ with
  $\bPP{\tilde{X}_J \tilde{Y}_J}$ and $\mathrm{P}_{XY}$ in the
  definition of $G_1(\bPP{\tilde{X}^n\tilde{Y}^n})$.}
\begin{align} 
G_1(\bPP{\tilde{X}^n\tilde{Y}^n}) \ge n G_1(\bPP{\tilde{X}_J
  \tilde{Y}_J}).
\label{eq:WZ-G1-bound}
\end{align}
For $G_2(\bPP{\tilde{U}\tilde{X}^n\tilde{Y}^n\tilde{Z}^n})$, note that
\begin{align}
- H(\tilde{X}^n|\tilde{U},\tilde{Y}^n) &= - \sum_{j=1}^n
H(\tilde{X}_j|\tilde{U},\tilde{X}_j^-,\tilde{Y}^n) \nonumber \\ &\ge -
\sum_{j=1}^n H(\tilde{X}_j|\tilde{U},\tilde{X}_j^-,\tilde{Y}_j^+,
Y_j), \nonumber \\ &= - n H(\tilde{X}_J | \tilde{U}_J,J,\tilde{Y}_J),
\nonumber
\end{align}
where $\tilde{U}_j = (\tilde{U},\tilde{X}_j^-,\tilde{Y}_j^+)$. Also,
$\mathbb{E}[d(\tilde{X}^n,\tilde{Z}^n)]
= n \mathbb{E}[ d(\tilde{X}_J, \tilde{Z}_J)].$
For the remaining terms in $G_2$, we have
\begin{align}
\lefteqn{ - H(\tilde{Y}^n|\tilde{U},\tilde{X}^n) + I(\tilde{Z}^n
  \wedge \tilde{X}^n|\tilde{U},\tilde{Y}^n)} \nonumber \\
&= - H(\tilde{X}^n|\tilde{U},\tilde{Y}^n,\tilde{Z}^n) +
H(\tilde{X}^n|\tilde{U}) - H(\tilde{Y}^n|\tilde{U}) \nonumber \\ &=
\sum_{j=1}^n \big[ - H(\tilde{X}_j |
  \tilde{U},\tilde{X}_j^-,\tilde{Y}^n,\tilde{Z}^n) + H(\tilde{X}_j |
  \tilde{U},\tilde{X}_j^-,\tilde{Y}_j^+) - H(\tilde{Y}_j |
  \tilde{U},\tilde{X}_j^-,\tilde{Y}_j^+) \big] \nonumber \\
&\geq n\big[ - H(\tilde{X}_J | \tilde{U}_J, J,
  \tilde{Y}_J,\tilde{Z}_J) + H(\tilde{X}_J | \tilde{U}_J,J) -
  H(\tilde{Y}_J | \tilde{U}_J,J) \big] \nonumber \\
&= n \big[ - H(\tilde{Y}_J|\tilde{U}_J,J,\tilde{X}_J) + I(\tilde{Z}_J
  \wedge \tilde{X}_J | \tilde{U}_J,J,\tilde{Y}_J) \big], \nonumber
\end{align}
where the second identity uses~\eqref{e:CKM_inequality}. Upon
combining the observations above, we get
\begin{align} 
G_2(\bPP{\tilde{U}\tilde{X}^n\tilde{Y}^n\tilde{Z}^n}) \ge n
G_2(\bPP{\tilde{U}_J J \tilde{X}_J \tilde{Y}_J \tilde{Z}_J}).
\label{eq:WZ-G2-bound}
\end{align}
Since \eqref{eq:WZ-G1-bound} and \eqref{eq:WZ-G2-bound} hold for an
arbitrary $\bPP{\tilde{U}\tilde{X}^n\tilde{Y}^n\tilde{Z}^n}$, the
proof is complete. \qed


\section{Interactive Function Computation Problem}\label{s:FC}
The second problem we consider entails the computation of a function
$f$ of $(X,Y)$ using interactive communication. For the ease of
presentation, we limit ourselves to protocols with $2$-rounds of
communication, but our analysis extends to protocols with bounded
(independent of $n$) rounds.

For a given source $\bPP{XY}$ on a finite alphabet ${\cal X}\times
{\cal Y}$, an ($2$-round) interactive communication protocol $\pi$
with inputs $(X^n,Y^n)$ consists of mappings $\varphi_1: \cX^n \to
\{0,1\}^{l_1}$ and $\varphi_2: \cY^n \times \{0,1\}^{l_1}\to
\{0,1\}^{l_2}$; the length of such a protocol $\pi$ is $l_1+l_2$. The
random transcript of the protocol is denoted by $\Pi = (\Pi_1, \Pi_2)$
where $\Pi_1 = \varphi_1(X^n)$, and $\Pi_2 = \varphi_2(Y^n, \Pi_1)$.

A protocol $\pi$ $\ep$-computes a function $f:\cX\times \cY\to \cZ$ if
we can form estimates $F_1^n = \psi_1(X^n, \Pi_2)$ and
$F_2^n=\psi_2(Y^n, \Pi_1)$ such that
\[
\bPr{ f(X_i, Y_i) = F_{1i} = F_{2i}, \, \forall i\in [n]} \geq
1-\ep.
\]
A rate $R>0$ is an $\ep$-achievable communication rate for $f$ if, for
all $n$ sufficiently large, there exists an interactive communication
protocol $\pi$ of length $|\pi|$ less than $nR$ that $\ep$-computes
$f$.  The infimum over all $\ep$-achievable communication rates for
$f$ is denoted by $R_f(\ep|\bPP{XY})$. The supremum over $\ep\in(0,1)$
of all $\ep$-achievable communication rates for $f$ is denoted by
$R_f(\bPP{XY})$.

The following characterization of ${\cal R}_{f}(\bPP{XY})$ was given
in~\cite{OrlRoc01}:\footnote{See~\cite{MaIsh11, BraRao11} for the
  extension to multiple rounds.}
\begin{align} \label{eq:R-interactive-function-computation}
R_f(\bPP{XY}) = \min I(U,V \wedge X|Y)+I(U,V\wedge Y|X),
\end{align}
where the minimum is over all $U, V$ satisfying $U \mc X\mc Y$, $V \mc
(Y,U)\mc X$; $H(f(X,Y)| Y,U) = H(f(X,Y)| X,U,V)=0$; and $|\cU|\leq
|\cX|, |\cV|\leq |\cY||\cX|$.
\begin{remark}
The right-side of \eqref{eq:R-interactive-function-computation} is
referred to as the {\em intrinsic information complexity} of $f$
(using $2$-round communication protocols) in \textchange{the} computer science
literature (cf.~\cite{BraRao11}). By noting the Markov relations $U
\mc X\mc Y$ and $V \mc (Y,U)\mc X$, we can obtain the following
equivalent expression for it:
\begin{align*}
I(U \wedge X|Y) + I(V \wedge Y|U,X),
\end{align*}
which is perhaps more commonly used in the information theory
literature (cf.~\cite{OrlRoc01, MaIsh11}).  In a similar vein, we use
the {\em extrinsic information complexity} $I(U, V \wedge X,Y)$ to
describe the results in Section \ref{sec:CR-SK}, which, too, can be
expressed alternatively using the aforementioned Markov relations.
\end{remark}

In the manner of Proposition~\ref{p:remove_Markov2}, we can replace
the ``hard'' Markov chain and functional constraints with ``soft''
divergence penalties to get (see Appendix
\ref{app:function-computation} for the proof)
\begin{align}
R_f(\bPP{XY}) = \sup_{\alpha>0} R_f^\alpha(\bPP{XY}),
\label{e:varitional_function}
\end{align}
where
\begin{align*}
 R_f^\alpha(\bPP{XY}) &:= \min_{\bPP{\tU \tV\tX\tY}} \big[
   I(\tU,\tV\wedge \tX|\tY)+I(\tU,\tV\wedge \tY|\tX) + \alpha
   D(\bPP{\tU \tV \tX \tY} \| \bQQ{\tU\tV XY})
   \\ &~~~+(\alpha+2)D(\bPP{\tX\tY} \| \bPP{XY}) + \alpha\big(
   D(\bPP{\tU|\tX\tY} \| \bPP{\tU|\tX} | \bPP{\tX\tY})
   +H(\tF|\tY,\tU)+ H(\tF|\tX, \tU, \tV) \big) \big] \\ &=
 \min_{\bPP{\tU \tV\tX\tY}} \big[ I(\tU,\tV\wedge
   \tX|\tY)+I(\tU,\tV\wedge \tY|\tX) +
   (2\alpha+2)D(\bPP{\tX\tY}\|\bPP{XY}) \\ &~~~ + \alpha\big( 2I(\tU
   \wedge \tY|\tX) + I(\tV\wedge \tX|\tY, \tU) +H(\tF|\tY,\tU)+
   H(\tF|\tX, \tU, \tV) \big) \big],
\end{align*}
$\tF= f(\tX,\tY)$, and $\bQQ{\tU \tV XY} = \bPP{\tV |\tU\tY}
\bPP{\tU|\tX}\bPP{XY}$ is the distribution induced from each $\bPP{\tU
  \tV\tX\tY}$ \textchange{and  respects the information constraints
of the coding problem.} The ranges $\cU$ and $\cV$ of $\tU$ and $\tV$ can be
restricted to $|\cU| \le |\cX||\cY|$ and $|\cV|\le |\cX|^2 |\cY|^2$.
The two forms described above have different utilities as described in
Remark~\ref{r:two_form}. As in the previous section, we divide the
proof of strong converse into two parts. The main technical component
is the following result.
\begin{theorem} \label{theorem:function_computation}
For every $n \in \mathbb{N}$ and $\alpha > 0$, we have
\begin{align*}
R_f^\alpha(\mathrm{P}_{XY}^n) \geq nR_f^{\alpha}(\bPP{XY}),
\end{align*}
where in defining $R_f^\alpha(\mathrm{P}_{XY}^n)$ we use $\tilde{F}^n
= (f(\tX_1,\tY_1), ..., f(\tX_n, \tY_n))$.
\end{theorem}
As corollary, we get the strong converse theorem for function
computation.
\begin{corollary} \label{corollary:function_computation}
For every $0 < \varepsilon < 1$, we have $ R_f(\varepsilon|\bPP{XY})=
R_f(\bPP{XY})$.
\end{corollary}
The proof of corollary follows from
Theorem~\ref{theorem:function_computation} using similar steps as the
proof of Corollary~\ref{corollary:WZ} where the changed measure
$\bPP{\tilde{X}^n\tilde{Y}^n}$ is now obtained by conditioning on the
set of inputs for which no error occurs, i.e., the set
\begin{align*}
{\cal D} &= \{ (x^n,y^n) : \psi_1(x^n, \varphi_2(y^n,
\varphi_1(x^n)))= \psi_2(y^n, \varphi_1(x^n)) = (f(x_1, y_1), ...,
f(x_n, y_n))\}.
\end{align*}
We close this section with a proof of
Theorem~\ref{theorem:function_computation}.

\paragraph*{Proof of Theorem~\ref{theorem:function_computation}} 
By setting
\begin{align*}
G_1(\bPP{\tX^n, \tY^n}) &:= \big[H(\tX^n|\tY^n)+
  D(\bPP{\tX^n\tY^n}\|\bPP{X^nY^n})\big] \\ &\qquad
+(2\alpha+1)\big[H(\tY^n|\tX^n)+
  D(\bPP{\tX^n\tY^n}\|\bPP{X^nY^n})\big]
\end{align*}
and
\begin{align*}
G_2(\bPP{\tilde{U}\tilde{V}\tilde{X}^n\tilde{Y}^n\tilde{F}^n}) &:= -
H(\tilde{X}^n|\tilde{Y}^n, \tilde U, \tilde V) -
H(\tilde{Y}^n|\tilde{X}^n, \tilde U, \tilde V) - 2\alpha
H(\tilde{Y}^n|\tilde{X}^n, \tilde U) \\ &\qquad\qquad + \alpha
I(\tilde V\wedge \tilde{X}^n|\tilde{Y}^n, \tilde U) + \alpha
H(\tilde{F}^n|\tilde{Y}^n, \tilde U) + \alpha
H(\tilde{F}^n|\tilde{X}^n, \tilde U, \tilde V)
\end{align*}
for given $\bPP{\tilde{U}\tilde{V}\tilde{X}^n\tilde{Y}^n}$ with
$\tilde{F}_i = f(\tilde{X}_i,\tilde{Y}_i)$ for $1 \le i \le n$, we can
write
\begin{align*}
R_f^\alpha(\mathrm{P}_{XY}^n) =
\min_{\bPP{\tilde{U}\tilde{V}\tilde{X}^n\tilde{Y}^n}} \big[
  G_1(\bPP{\tX^n, \tY^n}) +
  G_2(\bPP{\tilde{U}\tilde{V}\tilde{X}^n\tilde{Y}^n\tilde{F}^n})
  \big].
\end{align*} 
Fix arbitrary $\bPP{\tilde{U}\tilde{V}\tilde{X}^n\tilde{Y}^n}$.  By
Proposition \ref{proposition:almost-subadditivity-entropy}, we get
\begin{align}
G_1(\bPP{\tilde{X}^n\tilde{Y}^n}) \geq n G_1(\bPP{\tilde{X}_J
  \tilde{Y}_J}),
\label{e:function_superadd1}
\end{align}
where $J$ is distributed uniformly over $\{1, ...,n\}$.  For
$G_2(\bPP{\tilde{U}\tilde{V}\tilde{X}^n\tilde{Y}^n\tilde{F}^n})$,
since removing condition increases entropy, we have
\begin{align*}
- H(\tilde{X}^n|\tilde{Y}^n, \tilde U, \tilde V) -
H(\tilde{Y}^n|\tilde{X}^n, \tilde U, \tilde V) \geq
-nH(\tilde{X}_J|\tilde{Y}_J, \tilde{U}_J, J, \tilde{V})
-nH(\tilde{Y}_J|\tilde{X}_J, \tilde{U}_J, J, \tilde{V})
\end{align*}
where $\tilde{U}_j =(\tilde U, \tilde{X}_j^-,
\tilde{Y}_j^+)$. Furthermore, noting that
\begin{align*}
&- 2 H(\tilde{Y}^n|\tilde{X}^n, \tilde U) + I(\tilde V\wedge
  \tilde{X}^n|\tilde{Y}^n, \tilde U) + H(\tilde{F}^n|\tilde{Y}^n,
  \tilde U) + H(\tilde{F}^n|\tilde{X}^n, \tilde U, \tilde V) \\ &=
  2[H(\tilde{X}^n| \tilde{Y}^n,\tilde U) - H(\tilde{Y}^n|\tilde{X}^n,
    \tilde U)] +[H(\tilde{Y}^n| \tilde{X}^n, \tilde U, \tilde V) -
    H(\tilde{X}^n| \tilde{Y}^n, \tilde U, \tilde V)] \\&\quad -
  H(\tilde{X}^n| \tilde{Y}^n,\tilde U) - H(\tilde{Y}^n| \tilde{X}^n,
  \tilde U, \tilde V) + H(\tilde{F}^n|\tilde{Y}^n, \tilde U) +
  H(\tilde{F}^n|\tilde{X}^n, \tilde U, \tilde V) \\ &=
  2[H(\tilde{X}^n| \tilde{Y}^n,\tilde U) - H(\tilde{Y}^n|\tilde{X}^n,
    \tilde U)] +[H(\tilde{Y}^n| \tilde{X}^n, \tilde U, \tilde V) -
    H(\tilde{X}^n| \tilde{Y}^n, \tilde U, \tilde V)] \\&\quad -
  H(\tilde{X}^n, \tilde{F}^n | \tilde{Y}^n,\tilde U) - H(\tilde{Y}^n,
  \tilde{F}^n | \tilde{X}^n, \tilde U, \tilde V) +
  H(\tilde{F}^n|\tilde{Y}^n, \tilde U) + H(\tilde{F}^n|\tilde{X}^n,
  \tilde U, \tilde V) \\ &= 2[H(\tilde{X}^n| \tilde U) -
    H(\tilde{Y}^n| \tilde U)] + [H(\tilde{Y}^n| \tilde U, \tilde V)-
    H(\tilde{X}^n| \tilde U, \tilde V)] \\ &\quad -
  H(\tilde{X}^n|\tilde{Y}^n,\tilde{F}^n,\tilde{U}) - H(\tilde{Y}^n|
  \tilde{X}^n, \tilde{F}^n, \tilde U, \tilde V),
\end{align*}
where we used the fact that $\tilde{F}^n$ is function of
$(\tX^n,\tY^n)$ to append $\tF^n$ in the second equality. Thus, by
using~\eqref{e:CKM_inequality} twice, we obtain
\begin{align*}
&- 2 H(\tilde{Y}^n|\tilde{X}^n, \tilde U) + I(\tilde V\wedge
  \tilde{X}^n|\tilde{Y}^n, \tilde U) + H(\tilde{F}^n|\tilde{Y}^n,
  \tilde U) + H(\tilde{F}^n|\tilde{X}^n, \tilde U, \tilde V) \\ &\geq
  2n(H(\tilde{X}_J| \tilde U_J,J) - H(\tilde{Y}_J| \tilde U_J,J)) +
  n(H(\tilde{Y}_J| \tilde U_J, J, \tilde V)- H(\tilde{X}_J| \tilde
  U_J, J, \tilde V)) \\ &\quad - n
  H(\tilde{X}_J|\tilde{Y}_J,\tilde{F}_J,\tilde{U}_J, J ) -n
  H(\tilde{Y}_J| \tilde{X}_J, \tilde{F}_J, \tilde U_J, J, \tilde V)
  \\ &= - 2n H(\tilde{Y}_J|\tilde{X}_J, \tilde U_J, J) + nI(\tilde
  V\wedge \tilde{X}_J|\tilde{Y}_J, \tilde{U}_J, J) +
  nH(\tilde{F}_J|\tilde{Y}_J, \tilde{U}_J, J) + n
  H(\tilde{F}_J|\tilde{X}_J, \tilde{U}_J, J, \tilde V),
\end{align*}
where we used the fact $\tilde{F}_J = f(\tilde{X}_J,\tilde{Y}_J)$ to
remove the unnecessary $\tilde{F}_J$ in the previous identity. Upon
combining the bounds above, we obtain
\begin{align}
G_2(\bPP{\tilde{U}\tilde{V}\tilde{X}^n\tilde{Y}^n\tilde{F}^n}) \geq n
G_2(\bPP{(\tilde{U}_J, J) \tilde{V} \tX_J\tY_J \tF_J}).
 \label{e:function_superadd2}
\end{align}
Since \eqref{e:function_superadd1} and \eqref{e:function_superadd2}
hold for arbitrary $\bPP{\tU \tV \tX^n \tY^n}$, the proof is complete.
\qed

\section{Common Randomness Generation and Secret Key Agreement} \label{sec:CR-SK}

We now move to the closely related problems of common randomness
generation and secret key agreement. In these problems, an additional
challenge arises due to the presence of a total variation distance
constraint. We circumvent this difficulty by replacing the total
variation distance constraint by a constraint on log-likelihood; the
resulting set is used in our change of measure arguments. A similar
approach will be used later for the wiretap channel strong converse
where, too, the security constraint poses a similar challenge.

We begin with the common randomness problem and extend to the secret
key agreement case using the connection between the two problems.
Note that while the change of measure arguments presented here prove the
strong converse for secret key agreement with limited communication,
the strong converse for secret key agreement with unlimited
communication is available in \cite{TyaWat14ii}.  In fact, the
conditional independence testing bound of \cite{TyaWat14ii} yields
even the precise second-order term (cf.~\cite{HayTW14}); it is unclear
if our change of measure approach can do the same.

\subsection{Common Randomness Generation} \label{subsec:CR}

Consider a source $\bPP{XY}$ on a finite alphabet $\cX \times \cY$. An
($2$-round) interactive common randomness generation
protocol\footnote{For ease of presentation, we restrict to
  $2$-rounds. Our approach easily extends to higher (but fixed) number
  of rounds.} $\pi$ with input $(X^n,Y^n)$ consists of mappings
$\varphi_1:\cX^n \to \{0,1\}^{l_1}$ and $\varphi_2:\cY^n \times
\{0,1\}^{l_1} \to \{0,1\}^{l_2}$; the length $|\pi|$ of such a
protocol $\pi$ is $l_1 + l_2$. The random transcript of the protocol
is denoted by $\Pi = (\Pi_1,\Pi_2)$, where $\Pi_1 = \varphi_1(X^n)$
and $\Pi_2 = \varphi_2(Y^n, \Pi_1)$.

Given a protocol, a pair of random variables
$(K_1, K_2)$ taking 
values in a finite set $\cK$ constitute an $(\ep,\delta)$-CR
recoverable from $\pi$ if there exist $\psi_1:\cX^n \times
\{0,1\}^{l_2} \to \cK$ and $\psi_2:\cY^n \times \{0,1\}^{l_1} \to \cK$
such that $K_1 = \psi_1(X^n,\Pi_2)$, $K_2 = \psi_2(Y^n,\Pi_1)$, and
\begin{align}
\bPr{ K_1 \neq K_2 } &\le \ep, \label{eq:CR-error} \\ d(\bPP{K_1},
\bPP{\mathtt{unif}}) &\le \delta, \label{eq:CR-uniformity}
\end{align}  
where $\bPP{\mathtt{unif}}$ is the uniform distribution on $\cK$.  The
quantity $\log |\cK|$ denotes the length of the CR.

A rate pair $(R_\mathtt{c}, R_\mathtt{r})$ is
$(\ep,\delta)$-achievable if, for all $n$ sufficiently large, there
exists a protocol $\pi$ of length $|\pi| \le n R_\mathtt{c}$ that
recovers an $(\ep, \delta)$-CR of length $\log |\cK| \ge n
R_\mathtt{r}$. Let ${\cal R}_{\mathtt{CR}}(\ep,\delta|\bPP{XY})$ be
the closure of the set of all $(\ep,\delta)$-achievable rate
pairs. Define
\begin{align*}
{\cal R}_{\mathtt{CR}}(\bPP{XY}) := \bigcap_{0 < \ep,\delta <1} {\cal
  R}_{\mathtt{CR}}(\ep,\delta|\bPP{XY}).
\end{align*}
The following characterization of ${\cal R}_{\mathtt{CR}}(\bPP{XY})$
was given in \cite{AhlCsi98}:
\begin{align*}
{\cal R}_{\mathtt{CR}}(\bPP{XY}) &= \big\{ (R_\mathtt{c},R_\mathtt{r})
: \exists (U,V) \mbox{ s.t. } |\cU| \le |\cX|+1, |\cV| \le |\cX| |\cY|
+1, U \mc X \mc Y, V \mc (Y,U) \mc X \\ &~~~~~~~~~~~~ R_\mathtt{c} \ge
I(U,V \wedge X|Y) + I(U,V \wedge Y|X), R_\mathtt{r} \le I(U,V \wedge
X,Y) \big\}
\end{align*}
The set ${\cal R}_{\mathtt{CR}}(\bPP{XY})$ is closed and convex, and
it can be expressed alternatively using tangent lines as follows:
\begin{align*} 
{\cal R}_{\mathtt{CR}}(\bPP{XY}) = \bigcap_{\mu \ge 0} \big\{
(R_\mathtt{r}, R_\mathtt{c}) : R_\mathtt{r} - \mu R_\mathtt{c} \le
R_{\mathtt{CR}}^\mu(\bPP{XY}) \big\},
\end{align*}
where
\begin{align}
R_{\mathtt{CR}}^\mu(\bPP{XY}) &:= \max\big\{ I(U,V \wedge X,Y) -
\mu\big( I(U,V \wedge X|Y) + I(U,V \wedge Y|X) \big) : \nonumber
\\ &~~~~~~~~~~ \exists (U,V) \mbox{ s.t. } |\cU| \le |\cX|, |\cV| \le
|\cX| |\cY|, U \mc X \mc Y, V \mc (Y,U) \mc X \big\}.
 \label{eq:CR-supporting-line}
\end{align}
As before, we circumvent the Markov chain conditions by using the
following alternative form:
\begin{align} \label{eq:CR-variational-formula}
R_{\mathtt{CR}}^\mu(\bPP{XY}) = \inf_{\alpha > 0}
R_{\mathtt{CR}}^{\mu,\alpha}(\bPP{XY}),
\end{align}
where (see Remark~\ref{r:two_form})
\begin{align}
R_{\mathtt{CR}}^{\mu,\alpha}(\bPP{XY}) &:=
\max_{\bPP{\tilde{U}\tilde{V}\tilde{X}\tilde{Y}}} \big[
  I(\tilde{U},\tilde{V} \wedge \tilde{X},\tilde{Y}) - \mu \big(
  I(\tilde{U},\tilde{V} \wedge \tilde{X}|\tilde{Y}) +
  I(\tilde{U},\tilde{V} \wedge \tilde{Y}|\tilde{X}) \big) \nonumber
  \\ &~~~ - \alpha D(\bPP{\tU \tV \tX \tY} \| \bQQ{\tU\tV XY}) -
  D(\bPP{\tU\tV|\tX\tY} \| \bPP{\tU|\tX} \bPP{\tV|\tY\tU} |
  \bPP{\tX\tY}) - 2\mu D(\bPP{\tilde{X}\tilde{Y}} \| \bPP{XY}) \big]
\nonumber \\ &= \max_{\bPP{\tilde{U}\tilde{V}\tilde{X}\tilde{Y}}}
\big[ I(\tilde{U},\tilde{V} \wedge \tilde{X},\tilde{Y}) - \mu \big(
  I(\tilde{U},\tilde{V} \wedge \tilde{X}|\tilde{Y}) +
  I(\tilde{U},\tilde{V} \wedge \tilde{Y}|\tilde{X}) \big) \nonumber
  \\ &~~~ - (\alpha+2\mu) D(\bPP{\tilde{X}\tilde{Y}} \| \bPP{XY}) -
  (\alpha+1) \big( I(\tilde{U} \wedge \tilde{Y} | \tilde{X}) +
  I(\tilde{V} \wedge \tilde{X} | \tilde{Y},\tilde{U}) \big) \big],
\label{eq:CR-cost-form}
\end{align}
and $\bQQ{\tU \tV XY} = \bPP{\tV |\tU\tY} \bPP{\tU|\tX}\bPP{XY}$ is
the distribution induced from each $\bPP{\tU \tV\tX\tY}$.  The ranges
$\cU$ and $\cV$ of $\tU$ and $\tV$ can be restricted to $|\cU| \le
|\cX||\cY|$ and $|\cV|\le |\cX|^2 |\cY|^2$.

\begin{theorem} \label{thm:CR}
For every $n \in \mathbb{N}$, $\mu \ge 0$, and $\alpha > 0$, we have
\begin{align*}
R_{\mathtt{CR}}^{\mu,\alpha}(\mathrm{P}_{XY}^n) \le n
R_{\mathtt{CR}}^{\mu,\alpha}(\bPP{XY}).
\end{align*}
\end{theorem}

\begin{corollary} \label{cor:CR}
For every $0 < \ep,\delta < 1$ with $\ep + \delta < 1$, we have ${\cal
  R}_{\mathtt{CR}}(\ep,\delta|\bPP{XY}) = {\cal
  R}_{\mathtt{CR}}(\bPP{XY})$.
\end{corollary}

\paragraph*{Proof of Corollary \ref{cor:CR}}
For a given protocol $\pi$ and $(\ep, \delta)$-CR $(K_1, K_2)$
satisfying \eqref{eq:CR-error} and \eqref{eq:CR-uniformity}, we first
replace the uniformity constraint \eqref{eq:CR-uniformity} with a
constraint on log-likelihood. Specifically, for a given $\gamma >0$,
which will be specified later, let
\begin{align}
\cT_\gamma := \bigg\{ k \in \cK : \log \frac{1}{\bP{K_1}{k}} \ge \log
|\cK| - \gamma \bigg\}.
\label{eq:typical-set-CR}
\end{align}
Then, by \eqref{eq:CR-uniformity} and the standard argument in the
information-spectrum methods (cf.~\cite[Lemma 2.1.2]{Han03}), we have
\begin{align}
\delta &\ge d(\bPP{K_1}, \bPP{\mathtt{unif}}) \nonumber \\ &\ge
\bP{K_1}{\cT_\gamma^c} - \bP{\mathtt{unif}}{\cT_\gamma^c} \nonumber
\\ &\ge \bP{K_1}{\cT_\gamma^c} -
2^{-\gamma}. \label{eq:proof-cor-CR-1}
\end{align}
We now define the set $\cD$ over which our CR generation protocol
behaves ideally. Let
\begin{align}
\cD := \big\{ (x^n,y^n) : \psi_1(x^n,\varphi_2(y^n,\varphi_1(x^n)))
\in \cT_\gamma,~ \psi_1(x^n,\varphi_2(y^n,\varphi_1(x^n))) =
\psi_2(y^n, \varphi_1(x^n)) \big\}.
\label{eq:correct-set-CR}
\end{align}
By \eqref{eq:CR-error} and \eqref{eq:proof-cor-CR-1}, for $\gamma =
\log \frac{2}{1-\ep-\delta}$, we have
\begin{align}
\mathrm{P}_{XY}^n(\cD) &\ge 1 - \bP{K_1}{\cT_\gamma^c} - \bPr{ K_1
  \neq K_2 } \nonumber \\ &\ge 1 - \ep - \delta - 2^{-\gamma}
\nonumber \\ &= \frac{1 - \ep - \delta}{2}. \label{eq:proof-cor-CR-2}
\end{align}
Denote by $\bPP{\tilde{X}^n \tilde{Y}^n}$ the pmf
\begin{align}
\bP{\tilde{X}^n \tilde{Y}^n}{x^n,y^n} :=
\frac{\mathrm{P}_{XY}^n(x^n,y^n) \indicator[(x^n,y^n) \in \cD]
}{\mathrm{P}_{XY}^n(\cD)}.
\label{eq:changed-measure-CR}
\end{align}
Then, \eqref{eq:proof-cor-CR-2} implies
\begin{align}
D(\bPP{\tilde{X}^n \tilde{Y}^n} \| \mathrm{P}_{XY}^n) &= \log
\frac{1}{\mathrm{P}_{XY}^n(\cD)} \nonumber \\ &\le \log
\frac{2}{1-\ep-\delta}. \label{eq:proof-cor-CR-3}
\end{align}
Consider an execution of protocol $\pi$ for input
$(\tilde{X}^n,\tilde{Y}^n) \sim \bPP{\tilde{X}^n \tilde{Y}^n}$. Set
$\tilde{\Pi}_1 = \varphi_1(\tilde{X}^n)$, $\tilde{\Pi}_2 =
\varphi_2(\tilde{Y}^n,\tilde{\Pi}_1)$, $\tilde{\Pi} =
(\tilde{\Pi}_1,\tilde{\Pi}_2)$, $\tilde{K}_1 =
\psi_1(\tilde{X}^n,\tilde{\Pi}_2)$, and $\tilde{K}_2 =
\psi_2(\tilde{Y}^n,\tilde{\Pi}_1)$.  Note that $\tilde{K}_1 =
\tilde{K}_2$ with probability $1$. Furthermore, since the support of
$\bPP{\tilde{K}_1}$ satisfies $\mathtt{supp}(\bPP{\tilde{K}_1})
\subseteq \cT_\gamma$, we have
\begin{align*}
\bP{\tilde{K}_1}{k} &= \frac{1}{\mathrm{P}_{XY}^n(\cD)}
\sum_{(x^n,y^n) \in \cD : \atop
  \psi_1(x^n,\varphi_2(y^n,\varphi_1(x^n))) = k}
\mathrm{P}_{XY}^n(x^n,y^n) \\ &\le
\frac{\bP{K_1}{k}}{\mathrm{P}_{XY}^n(\cD)} \\ &\le
\frac{2^\gamma}{\mathrm{P}_{XY}^n(\cD) |\cK|},
\end{align*}
for every $k \in \mathtt{supp}(\bPP{\tilde{K}_1})$. Thus, we get
\begin{align*}
H_{\min}(\bPP{\tilde{K}_1}) &= \min_{k \in
  \mathtt{supp}(\bPP{\tilde{K}_1})} \log \frac{1}{\bP{\tilde{K}_1}{k}}
\\ &\ge \log |\cK| - 2 \log \frac{2}{1-\ep-\delta},
\end{align*}
where we used \eqref{eq:proof-cor-CR-2} once again in the inequality.

By noting $H_{\min}(\bPP{\tilde{K}_1}) \le H(\tilde{K}_1)$, we have
\begin{align*}
 n (R_\mathtt{r}- \mu R_\mathtt{c}) - 2 \log \frac{2}{1-\ep-\delta}
 &\le \log |\cK| - \mu |\pi| - 2 \log \frac{2}{1-\ep-\delta} \\ &\le
 H(\tilde{K}_1) - \mu H(\tilde{\Pi}) \\ &\le
 H(\tilde{\Pi},\tilde{K}_1) - \mu \big( H(\tilde{\Pi}|\tilde{X}^n) +
 H(\tilde{\Pi}|\tilde{Y}^n) \big) \\ &= H(\tilde{\Pi},\tilde{K}_2) -
 \mu \big( H(\tilde{\Pi},\tilde{K}_2|\tilde{X}^n) +
 H(\tilde{\Pi},\tilde{K}_2|\tilde{Y}^n) \big) \\ &\le
 I(\tilde{\Pi},\tilde{K}_2 \wedge \tilde{X}^n,\tilde{Y}^n) - \mu \big(
 I(\tilde{\Pi},\tilde{K}_2 \wedge \tilde{X}^n | \tilde{Y}^n) +
 I(\tilde{\Pi}, \tilde{K}_2 \wedge \tilde{Y}^n| \tilde{X}^n) \big)
 \\ &\le I(\tilde{\Pi},\tilde{K}_2 \wedge \tilde{X}^n,\tilde{Y}^n) -
 \mu \big( I(\tilde{\Pi},\tilde{K}_2 \wedge \tilde{X}^n | \tilde{Y}^n)
 + I(\tilde{\Pi}, \tilde{K}_2 \wedge \tilde{Y}^n| \tilde{X}^n) \big)
 \\ &~~~ - (\alpha+1)\big( I(\tilde{\Pi}_1 \wedge
 \tilde{Y}^n|\tilde{X}^n) + I(\tilde{\Pi}_2, \tilde{K}_2 \wedge
 \tilde{X}^n|\tilde{Y}^n,\tilde{\Pi}_1) \big) \\ &~~~- (\alpha+2\mu)
 D(\bPP{\tilde{X}^n \tilde{Y}^n} \| \mathrm{P}_{XY}^n) + (\alpha+2\mu)
 \log \frac{2}{1-\ep-\delta} \\ &\le
 R_{\mathtt{CR}}^{\mu,\alpha}(\mathrm{P}_{XY}^n) + (\alpha+2\mu) \log
 \frac{2}{1-\ep-\delta},
\end{align*}
where we used a well-known property of interactive communication in
the third inequality (eg.~see \cite[Eq.~(3.2)]{NarayanTyagi16}); the
identity follows from the fact that $\tilde{K}_1$ and $\tilde{K}_2$
are recoverable perfectly from $(\tilde{X}^n,\tilde{\Pi})$ and
$(\tilde{Y}^n,\tilde{\Pi})$, respectively, and
$\tilde{K}_1=\tilde{K}_2$ with probability $1$; and we used the fact
that costs $I(\tilde{\Pi}_1 \wedge \tilde{Y}^n|\tilde{X}^n)$ and
$I(\tilde{\Pi}_2, \tilde{K}_2 \wedge
\tilde{X}^n|\tilde{Y}^n,\tilde{\Pi}_1)$ are both $0$ and
\eqref{eq:proof-cor-CR-3} in the fifth inequality; in the last
inequality, we regarded $\tilde{\Pi}_1$ and
$(\tilde{\Pi}_2,\tilde{K}_2)$ as $\tilde{U}$ and $\tilde{V}$,
respectively.  Finally, by applying Theorem \ref{thm:CR}, we have
\begin{align*}
R_\mathtt{r} - \mu R_\mathtt{c} \le
R_{\mathtt{CR}}^{\mu,\alpha}(\bPP{XY}) + \frac{(\alpha + 2\mu + 2)}{n}
\log \frac{2}{1-\ep-\delta},
\end{align*}
which together with \eqref{eq:CR-variational-formula} imply the strong
converse. \qed

\paragraph*{Proof of Theorem \ref{thm:CR}}
First note that we can expand
\begin{align*}
I(\tU,\tV \wedge \tX^n, \tY^n) = I(\tU \wedge \tX^n) + I(\tV \wedge
\tY^n | \tU) + I(\tU \wedge \tY^n | \tX^n) + I(\tV \wedge \tX^n |
\tY^n, \tU).
\end{align*}
Then, by setting
\begin{align*}
G_1(\bPP{\tX^n \tY^n}) := H(\tX^n) - \mu \big[ H(\tX^n | \tY^n) +
  D(\bPP{\tX^n \tY^n} \| \mathrm{P}_{XY}^n) \big] - (\alpha+\mu) \big[
  H(\tY^n | \tX^n) + D(\bPP{\tX^n \tY^n} \| \mathrm{P}_{XY}^n) \big]
\end{align*}
and
\begin{align*}
G_2(\bPP{\tU \tV \tX^n \tY^n}) &:= - H(\tX^n | \tU) + I(\tV \wedge
\tY^n | \tU) + \mu \big( H(\tX^n | \tY^n,\tU, \tV) + H(\tY^n | \tX^n,
\tU, \tV) \big) \\ &~~~~~~~~~~~~~~~ + \alpha \big( H(\tY^n | \tX^n,
\tU) - I(\tV \wedge \tX^n | \tY^n, \tU) \big)
\end{align*}
for given $\bPP{\tilde{U}\tilde{V}\tilde{X}^n\tilde{Y}^n}$, we can
write
\begin{align*}
R_{\mathtt{CR}}^{\mu,\alpha}(\mathrm{P}_{XY}^n) =
\max_{\bPP{\tilde{U}\tilde{V}\tilde{X}^n\tilde{Y}^n}} \big[
  G_1(\bPP{\tX^n \tY^n}) + G_2(\bPP{\tU \tV \tX^n \tY^n}) \big].
\end{align*}
Fix arbitrary $\bPP{\tU \tV \tX^n \tY^n}$.  By noting $H(\tX^n) \le n
H(\tX_J)$ and by using Proposition
\ref{proposition:almost-subadditivity-entropy}, we get
\begin{align} \label{eq:proof-the-CR-1}
G_1(\bPP{\tX^n \tY^n}) \le n G_1(\bPP{\tX_J \tY_J}),
\end{align}
where $J$ is distributed uniformly over $\{1,\ldots,n\}$. For
$G_2(\bPP{\tU \tV\tX^n \tY^n})$, by using~\eqref{e:CKM_inequality}, we
have
\begin{align*}
- H(\tX^n | \tU) + I(\tV \wedge \tY^n | \tU) &= H(\tY^n | \tU) -
H(\tX^n |\tU) - H(\tY^n | \tU, \tV) \\ &\le n \big[ H(\tY_J | \tU_J,J)
  - H(\tX_J | \tU_J, J) - H(\tY_J | \tU_J,J, \tV) \big],
\end{align*}
where $\tU_j = (\tU,\tX_j^-, \tY_j^+)$. Also,
\begin{align*}
H(\tX^n | \tY^n,\tU, \tV) + H(\tY^n | \tX^n, \tU, \tV) \le n \big(
H(\tX_J | \tY_J, \tU_J, J,\tV) + H(\tY_J | \tX_J, \tU_J, J, \tV)
\big).
\end{align*}
Furthermore, by using~\eqref{e:CKM_inequality} once more, we obtain
\begin{align*}
H(\tY^n | \tX^n, \tU) - I(\tV \wedge \tX^n | \tY^n, \tU) &= H(\tY^n |
\tU) - H(\tX^n| \tU) + H(\tX^n | \tY^n,\tU,\tV) \\ &\le n \big(
H(\tY_J | \tU_J,J) - H(\tX_J | \tU_J,J) + H(\tX_J | \tY_J, \tU_J,J,
\tV) \big) \\ &= n \big( H(\tY_J | \tX_J, \tU_J,J) - I(\tV \wedge
\tX_J | \tY_J, \tU_J,J) \big).
\end{align*}
Upon combining the bounds above, we obtain
\begin{align} \label{eq:proof-the-CR-2}
G_2(\bPP{\tU \tV \tX^n \tY^n}) \le n G_2(\bPP{\tU_J J \tV \tX_J
  \tY_J}).
\end{align}
Since \eqref{eq:proof-the-CR-1} and \eqref{eq:proof-the-CR-2} hold for
arbitrary $\bPP{\tU\tV \tX^n \tY^n}$, the proof is complete. \qed

\begin{remark}
When randomization is allowed, the achievable region is given by
\begin{align*}
\tilde{{\cal R}}_{\mathtt{CR}}(\bPP{XY}) = \big\{ (R_\mathtt{c},
R_\mathtt{r}) : \exists\, t \ge 0 \mbox{ s.t. } (R_\mathtt{c}-t,
R_\mathtt{r} -t ) \in {\cal R}_{\mathtt{CR}}(\bPP{XY}) \big\}.
\end{align*}
We can extend the proof above to randomized protocols easily by
appending two independent i.i.d. sources $A^n$ and $B^n$ (taking
values in sufficiently large alphabets $\cA$ and $\cB$) to $X^n$ and
$Y^n$, respectively. By Corollary \ref{cor:CR} we obtain ${\cal
  R}_{\mathtt{CR}}(\ep,\delta| \bPP{XAYB}) = {\cal
  R}_{\mathtt{CR}}(\bPP{XAYB})$. Also, noting that
(cf.~\cite{AhlCsi98})
\begin{align*}
\bigcup_{\bPP{AB}} {\cal R}_{\mathtt{CR}}(\bPP{XAYB}) = \tilde{{\cal
    R}}_{\mathtt{CR}}(\bPP{XY}),
\end{align*}
where the union is taken over all distributions such that $A$ and $B$
are independent, we have the strong converse even with randomized
protocols. \textchange{A similar approach has been pursued in~\cite[proof of Theorem
      III.2]{SudanTyagiWatanabe19} to handle randomization.}
\end{remark}

\subsection{Secret Key Agreement} \label{subsec:SK}
Next, we consider the secret key agreement problem. The formulation
and analysis is very similar to the common randomness generation
problem; we only highlight the differences. Specifically, an $(\ep,
\delta)$-SK $(K_1,K_2)$ recoverable from a protocol $\pi$ is an $(\ep,
\delta)$-CR with the uniformity condition \eqref{eq:CR-uniformity}
replaced by the secrecy condition
\begin{align} \label{eq:SK-security}
d(\bPP{K_1 \Pi}, \bPP{\mathtt{unif}} \times \bPP{\Pi}) \le \delta.
\end{align}
An $(\ep, \delta)$-achievable secret key rate pair
$(R_\mathtt{c},R_\mathtt{s})$ and the rate regions ${\cal
  R}_{\mathtt{SK}}(\ep,\delta | \bPP{XY})$ and ${\cal
  R}_{\mathtt{SK}}(\bPP{XY})$ are defined exactly as before.  The
following characterization of ${\cal R}_{\mathtt{SK}}(\bPP{XY})$ was
given in \cite{Tya13}:
\begin{align*}
{\cal R}_{\mathtt{SK}}(\bPP{XY}) &= \big\{ (R_\mathtt{c},R_\mathtt{s})
: \exists (U,V) \mbox{ s.t. } |\cU| \le |\cX|+1, |\cV| \le |\cX| |\cY|
+1, U \mc X \mc Y, V \mc (Y,U) \mc X \\ &~~~~~~~~~~~~ R_\mathtt{c} \ge
I(U,V \wedge X|Y) + I(U,V \wedge Y|X), \\ &~~~~~~~~~~~~ R_\mathtt{s}
\le I(U,V\wedge X,Y) - I(U,V \wedge X|Y) - I(U,V \wedge Y|X) \big\}.
\end{align*}
Define $R_{\mathtt{SK}}^\mu(\bPP{XY})$ and
$R_{\mathtt{SK}}^{\mu,\alpha}(\bPP{XY})$ analogously to
\eqref{eq:CR-supporting-line} and \eqref{eq:CR-cost-form},
respectively. Then, it can be easily verified that
\begin{align*}
R_{\mathtt{SK}}^\mu(\bPP{XY}) &= R_{\mathtt{CR}}^{\mu+1}(\bPP{XY}),
\\ R_{\mathtt{SK}}^{\mu,\alpha}(\bPP{XY}) &=
R_{\mathtt{CR}}^{\mu+1,\alpha}(\bPP{XY}).
\end{align*}
Thus, Theorem \ref{thm:CR} can be rewritten as follows.
\begin{theorem} \label{thm:SK}
For every $n \in \mathbb{N}$, $\mu \ge 0$, and $\alpha > 0$, we have
\begin{align*}
R_{\mathtt{SK}}^{\mu,\alpha}(\mathrm{P}_{XY}^n) \le n
R_{\mathtt{SK}}^{\mu,\alpha}(\bPP{XY}).
\end{align*}
\end{theorem}

Furthermore, by using Theorem \ref{thm:SK}, we have the following
strong converse theorem.
\begin{corollary} \label{cor:SK}
For every $0 < \ep, \delta < 1$ with $\ep + \delta < 1$, we have
${\cal R}_{\mathtt{SK}}(\ep,\delta | \bPP{XY}) = {\cal
  R}_{\mathtt{SK}}(\bPP{XY})$.
\end{corollary}

\paragraph*{Proof of Corollary \ref{cor:SK}}

The proof is mostly the same as the proof of Corollary \ref{cor:CR};
we only highlight the modifications required.  Instead of the set
defined by \eqref{eq:typical-set-CR}, we consider
\begin{align*}
\cT_\gamma := \bigg\{ (k,\tau) : \log \frac{1}{\bP{K_1|\Pi}{k|\tau}}
\ge \log |\cK| - \gamma \bigg\}.
\end{align*}
We can verify that
\begin{align*}
\delta \ge \bP{K_1 \Pi}{\cT_\gamma^c} - 2^{-\gamma}.
\end{align*}
In place of \eqref{eq:correct-set-CR}, define the set $\cD$ as
\begin{align*}
\cD &:= \big\{ (x^n,y^n) : (
\psi_1(x^n,\varphi_2(y^n,\varphi_1(x^n))), \varphi_1(x^n),
\varphi_2(y^n,\varphi_1(x^n))) \in \cT_\gamma, \\ &\hspace{5.5cm}
\psi_1(x^n,\varphi_2(y^n,\varphi_1(x^n))) = \psi_2(y^n,\varphi_1(x^n))
\big\}.
\end{align*}
Then, for the changed measure \eqref{eq:changed-measure-CR}, we
recover the bound \eqref{eq:proof-cor-CR-3}. Also,
\begin{align*}
H_{\min}(\bPP{\tilde{K}_1\tilde{\Pi}} | \bPP{\tilde{\Pi}}) &:=
\min_{(k,\tau) \in \mathtt{supp}(\bPP{\tilde{K}_1 \tilde{\Pi}})} \log
\frac{1}{\bP{\tilde{K}_1|\tilde{\Pi}}{k|\tau}} \\ &\ge \log |\cK| - 2
\log \frac{2}{1-\ep-\delta}.
\end{align*}
Therefore, upon noting $H_{\min}(\bPP{\tilde{K}_1\tilde{\Pi}} |
\bPP{\tilde{\Pi}}) \le H(\tilde{K}_1|\tilde{\Pi})$, we obtain
\begin{align*}
n(R_\mathtt{s} - \mu R_\mathtt{c}) - 2 \log \frac{2}{1-\ep-\delta}
&\le H(\tilde{K}_1|\tilde{\Pi}) - \mu H(\tilde{\Pi}) \\ &=
H(\tilde{\Pi},\tilde{K}_1) - (\mu+1) H(\tilde{\Pi}) \\ &\le
R_{\mathtt{SK}}^{\mu,\alpha}(\mathrm{P}_{XY}^n) +
(\alpha+2\mu+2)\log\frac{2}{1-\ep-\delta}.
\end{align*}
Finally, the strong converse follows from Theorem \ref{thm:SK}.  \qed

\section{Wiretap Channel} \label{sec:wiretap}

A wiretap channel code enables reliable transmission of a message over
a noisy channel while keeping it secure from an eavesdropper who can
see another noisy version of transmissions. Formally, given a discrete
memoryless channel (DMC) $W: \cX\to \cY \times \cZ$, an $(N, n, \ep,
\delta)$-wiretap code for $W$ consists of a (possibly randomized)
encoder $\varphi: \{1, ..., N\} \to \cX^n$ and a decoder $\psi: \cY^n
\to \{1, ..., N\}$ such that when a message $M$ distributed uniformly
over $\{1, ..., N\}$ is transmitted over the channel as $X^n =
\varphi(M)$, the estimate $\hat M = \psi(Y^n)$ has probability of
error satisfying $\mathbb{P}(\hat{M} \neq M)\leq \ep$ and leakage
$d(\bPP{MZ^n}, \bPP{M}\times \bPP{Z^n})\leq \delta$.

A rate $R>0$ is $(\ep, \delta)$-achievable if there exists an
$(\lfloor 2^{nR}\rfloor, n, \ep, \delta)$-wiretap code for all $n$
sufficiently large. The $(\ep, \delta)$-wiretap capacity
$C_\mathtt{s}(\ep, \delta|W)$ is the supremum over all
$(\ep,\delta)$-achievable rates. The wiretap capacity
$C_\mathtt{s}(W)$ is the infimum of $C_\mathtt{s}(\ep, \delta|W)$ over
all $\ep, \delta \in (0,1)$. The following characterization of
$C_\mathtt{s}(W)$ was derived in~\cite{CsiKor78}:
\[
C_\mathtt{s}(W) = \max_{\bPP{UXYZ}: \atop{\bPP{YZ|XU} = W}} \big[ I(U
  \wedge Y) - I(U\wedge Z) \big],
\]
where the cardinality $|\cU|$ of $U$ can be restricted to be $|\cU|
\le |\cX|$.  Using Proposition~\ref{p:remove_Markov2}, the expression
on the right above can be written alternatively as\footnote{The proof
  of \eqref{e:wiretap_variational} is very similar to that of other
  variational formulae such as \eqref{eq:WZ-variational-form} and is
  omitted.}
\begin{align}
C_\mathtt{s}(W) = \inf_{\alpha>0}C_\mathtt{s}^\alpha(W),
\label{e:wiretap_variational}
\end{align}
where
\[
C_\mathtt{s}^\alpha(W) = \max_{\bPP{\tU\tX\tY\tZ}} \big[ I(\tU \wedge
  \tY) - I(\tU\wedge \tZ) - \alpha D(\bPP{\tY\tZ|\tX\tU}\| W|
  \bPP{\tX\tU}) \big],
\]
where the cardinality $|\cU|$ of $U$ can be restricted to be $|\cU|
\le |\cX||\cY||\cZ|$.  The next theorem shows that the quantity
$C_\mathtt{s}^\alpha(W)$ satisfies the required sub-additivity
property.
\begin{theorem}\label{t:wiretap_subadditivity}
Consider a DMC $W:\cX \to \cY \times \cZ$ such that
$W(y,z|x)=W_1(y|x)W_2(z|x)$. For every $n\in \mN$ and $\alpha >0$,
\[
C_\mathtt{s}^{2\alpha}(W^n) \leq n C_\mathtt{s}^{\alpha}(W).
\]
\end{theorem}

As a corollary, we obtain the strong converse for wiretap
channel.\footnote{\textchange{We remark that we consider strong
    converse for the wiretap channel
   only when information leakage is measured by $d(\bPP{MZ^n},
   \bPP{M}\times \bPP{Z^n})$, and not for other measures of secrecy
   such as those considered in~\cite{BlochLaneman13}.}}
\begin{corollary}\label{c:wiretap_strong_converse}
For every $0 < \ep, \delta < 1$ with $\ep+\delta <1$, we have
$C_\mathtt{s}(\ep, \delta|W) = C_\mathtt{s}(W)$.
\end{corollary}

\paragraph*{Proof of Corollary~\ref{c:wiretap_strong_converse}}
Consider an $(N, n, \ep, \delta)$-wiretap code with a randomized
encoder $\varphi$ and (deterministic) decoder $\psi$ . Note that
without loss of generality we may assume $W(y,z|x) = W_1(y|x)W_2(z|x)$
since the error and secrecy criterion, respectively, depend only on
the marginals $(X^n, Y^n)$ and $(X^n, Z^n)$. The first step in our
proof is to convert the average probability of error and secrecy
requirements to a worst-case version. Specifically, since
\begin{align*}
\ep + \delta & \geq \mathbb{P}(\hat{M} \neq M) + d(\bPP{M Z^n},
\bPP{M} \times \bPP{Z^n}) \\ &\textchange{=} \frac{1}{N} \sum_{m=1}^N \bigg[
  \mathbb{P}( \hat M \neq m | M=m) + d(\bPP{Z^n|M=m}, \bPP{Z^n})
  \bigg] ,
\end{align*}
there exists a subset $\cM^\prime$ of size $|\cM^\prime | \geq
(1-\ep-\delta)N/(1+\ep+\delta)$ such that for every message $m \in
\cM^\prime$,
\[
\mathbb{P}( \hat{M} \neq m | M=m) + d(\bPP{Z^n|M=m}, \bPP{Z^n}) \leq
\frac {1 + \ep+ \delta}2.
\]
For $m \in {\cal M}^\prime$, consider the sets
\[
\cA_m = \{y^n: \psi(y^n)=m\}
\]
and, for $\gamma>0$ specified later,
\[
\cB_m = \left\{z^n: \log \frac {\bP{Z^n|M}{z^n|m} }{\bP{Z^n}{z^n}}\leq
\gamma\right\}.
\]
The set $\cB_m$ denotes, roughly, the set of observations that do not
reveal much information to the wiretapper about the message $m$ -- the
wiretapper cannot distinguish reliably if the observation was
generated from $\bPP{Z^n|M=m}$ or from $\bPP{Z^n}$. By the standard
argument in the information-spectrum methods (cf.~\cite{Han03}), the
set $\cB_m$ satisfies
\[
\bP{Z^n}{\cB_m^c} \leq 2^{-\gamma}.
\]
Furthermore, from the definition of the total variation distance, we
further have
\begin{align*}
\bP{Z^n|M=m}{\cB_m^c} &\leq 2^{-\gamma} + d(\bPP{Z^n|M=m}, \bPP{Z^n})
\end{align*}
for every $m \in {\cal M}^\prime$.  Therefore, upon choosing
$2^{-\gamma} = (1-\ep-\delta)/4$, we have
\begin{align*}
\bPr{Y^n \in \cA_m, Z^n \in \cB_m | M=m} &\geq 1 -
\bP{Y^n|M=m}{\cA_m^c} - \bP{Z^n|M=m}{\cB_m^c} \\ &\geq
\frac{1-\ep-\delta}4
\end{align*}
for every $m \in {\cal M}^\prime$.  Denote $\eta = 1-
(1-\ep-\delta)/4$ and by $\cC_m$ the set of $x^n \in
\mathtt{supp}(\bPP{X^n|M=m})$ such that
\begin{align}
\bPr{Y^n \in \cA_m, Z^n \in \cB_m | X^n=x^n} \geq 1- \sqrt{\eta},
\label{e:m-set}
\end{align}
which satisfies
\begin{align}
\bPr{X^n \in \cC_m|M=m}\geq 1- \sqrt{\eta}
\label{e:x-set}
\end{align}
by the reverse Markov inequality.  We now define our modified random
variables for which the code is perfectly error-free and has a small
leakage of information to the wiretapper; however, unlike the original
random variables satisfying the Markov constraint $M \mc X^n \mc
(Y^n,Z^n)$, the modified random variables do not satisfy the Markov
constraint (see also Remark \ref{remark:wiretap-markov}).
Specifically, consider random variables $(\tU, \tX^n, \tY^n, \tZ^n)$
such that $\tU$ is uniformly distributed on ${\cal M}^\prime$, and
\begin{align*}
\bP{\tX^n|\tU}{x^n| m}&= \frac{\bP{X^n|M}{x^n|m}\indicator[x^n\in
    \cC_m]} {\bP{X^n|M}{\cC_m|m}}; \\ \bP{\tY^n\tZ^n|\tX^n\tU}{y^n,
  z^n|x^n, m}&= \frac{\bP{Y^nZ^n|X^n}{y^n,z^n|x^n}\indicator[y^n\in
    \cA_m, z^n\in \cB_m]} {\bP{Y^nZ^n|X^n}{\cA_m\times \cB_m|x^n}},
\quad \, \forall x^n \in \cC_m.
\end{align*}
Using the conditional independence assumption $W(y,z|x) =
W_1(y|x)W_2(z|x)$, we further get that
\begin{align*}
\bP{\tY^n\tZ^n|\tX^n\tU}{y^n, z^n|x^n, m}&=
\frac{\bP{Y^n|X^n}{y^n|x^n}\indicator[y^n\in \cA_m]}
     {\bP{Y^n|X^n}{\cA_m|x^n}} \cdot
     \frac{\bP{Z^n|X^n}{z^n|x^n}\indicator[z^n\in \cB_m]}
          {\bP{Z^n|X^n}{\cB_m|x^n}},
\end{align*}
whereby
\begin{align}
\bP{\tZ^n|\tX^n\tU}{z^n|x^n,m}&=
\frac{\bP{Z^n|X^n}{z^n|x^n}\indicator[z^n\in \cB_m]}
     {\bP{Z^n|X^n}{\cB_m|x^n}}.
\label{e:z-marginal}
\end{align}
Since $\tU = \psi(\tY^n)$ with probability $1$, we get
\begin{align}
\log N - \log \frac 2{1-\ep-\delta} \leq \log |{\cal M}^\prime| \leq
I(\tU \wedge \tY^n).
\label{e:rate_bound1}
\end{align}
To bound the leakage $I(\tU \wedge \tZ^n)$, note that
\begin{align*}
I(\tU \wedge \tZ^n) + D(\bPP{\tZ^n}\|\bPP{Z^n}) &= \mathbb{E}\bigg[
  \log \frac{\mathrm{P}_{\tZ^n|\tU}(\tZ^n|\tU)}{\mathrm{P}_{Z^n}(
    \tZ^n)} \bigg] \\ &\leq \max_{m \in {\cal M}^\prime, z^n \in
  \cB_m} \log \frac{\bP{\tZ^n|\tilde{U}}{z^n|m}}{\bP{Z^n}{ z^n}}
\\ &\leq \log \frac 4{1-\ep-\delta} + \max_{m \in {\cal M}^\prime, z^n
  \in \cB_m} \log
\frac{\bP{\tZ^n|\tilde{U}}{z^n|m}}{\bP{Z^n|M}{z^n|m}},
\end{align*}
where the previous inequality uses the definition of $\cB_m$ and
$\gamma = \log \frac{4}{1-\ep-\delta}$. For the second term on the
right-side above, for every $z^n \in \cB_m$ it holds that
\begin{align*}
\frac{\bP{\tZ^n|\tilde{U}}{z^n|m}}{\bP{Z^n|M}{z^n|m}} &=
\frac{\sum_{x^n\in \cC_m} \bP{\tX^n|\tilde{U}}{x^n|m}\bP{\tZ^n|\tX^n
    \tilde{U}}{z^n|x^n,m}} {\sum_{x^{\prime n}\in \cX^n}
  \bP{X^n|M}{x^{\prime n}|m}\bP{Z^n|X^n}{z^n|x^{\prime n}}} \\ &\leq
\frac{\sum_{x^n\in \cC_m} \bP{\tX^n|\tilde{U}}{x^n|m}\bP{\tZ^n|\tX^n
    \tilde{U}}{z^n|x^n,m}} {\sum_{x^{\prime n}\in \cC_m}
  \bP{X^n|M}{x^{\prime n}|m}\bP{Z^n|X^n}{z^n|x^{\prime n}}} \\ &=
\frac{\sum_{x^n\in \cC_m}\bP{X^n|M}{x^n|m}\bP{Z^n|X^n}{z^n|x^n} / \{
  \bP{X^n|M}{\cC_m|m}\bP{Z^n|X^n}{\cB_m|x^n} \}} {\sum_{x^{\prime
      n}\in \cC_m} \bP{X^n|M}{x^{\prime
      n}|m}\bP{Z^n|X^n}{z^n|x^{\prime n}}} \\ &\leq
\frac{1}{(1-\sqrt{\eta})^2},
\end{align*}
where the second equality uses \eqref{e:z-marginal} and the final
inequality is by \eqref{e:m-set} and \eqref{e:x-set}.  Using the
bounds above, we get
\[
I(\tU \wedge \tZ^n) \leq \log \frac 4{1-\ep-\delta} + 2 \log \frac
1{1-\sqrt{\eta}}.
\]
Combining this bound with \eqref{e:rate_bound1}, for every $\alpha>0$
and with
\[
\Delta(\ep, \delta)= 2 \log \frac 1 {1-\ep-\delta} + 2\log \frac 1 {1-
  \sqrt{1- (1-\ep-\delta)/4}} + 3,
\]
we get
\begin{align}
 \log N &\leq I(\tU \wedge \tY^n)-I(\tU \wedge \tZ^n) +\Delta(\ep,
 \delta) \nonumber \\ &\leq C_\mathtt{s}^{2\alpha}(W^n) + 2\alpha
 D(\bPP{\tY^n\tZ^n|\tX^n\tU}\| W^n| \bPP{\tX^n\tU}) +\Delta(\ep,
 \delta) \nonumber \\ &\leq nC_\mathtt{s}^{\alpha}(W) + 2\alpha
 D(\bPP{\tY^n\tZ^n|\tX^n\tU}\| W^n| \bPP{\tX^n\tU}) +\Delta(\ep,
 \delta),
\label{e:rate_bound2}
\end{align}
where the final bound uses Theorem~\ref{t:wiretap_subadditivity}. It
only remains to bound the divergence term on the right-side above.  To
that end, note
\begin{align*}
D(\bPP{\tY^n\tZ^n|\tX^n\tU}\| W^n| \bPP{\tX^n\tU}) &= \sum_{x^n,m}
\bP{\tilde{X}^n \tilde{U}}{x^n,m} \log \frac
1{\bP{Y^nZ^n|X^n}{\cA_m\times \cB_m|x^n}} \\ &\leq \log \frac
1{1-\sqrt{\eta}},
\end{align*}
where we have used the fact that support of $\bPP{\tX^n|\tU=m}$ is
$\cC_m$ and \eqref{e:m-set}.  This bound along with
\eqref{e:rate_bound2} yields
\begin{align*}
\log N \leq nC_{\alpha}(W) + 2 \log \frac 1 {1-\ep-\delta} +
(2\alpha+2)\log \frac 1 {1- \sqrt{1- (1-\ep-\delta)/4}} + 3,
\end{align*}
which yields the strong converse by \eqref{e:wiretap_variational}.
\qed

\begin{remark} \label{remark:wiretap-markov}
Unlike the standard choice of auxiliary random variable in the wiretap
channel, the random variables
$(\tilde{U},\tilde{X}^n,\tilde{Y}^n,\tilde{Z}^n)$ in the above proof
do not satisfy the Markov relation $\tU \mc \tX^n \mc (\tY^n, \tZ^n)$.
Instead, we have added the cost $D(\bPP{\tilde{Y}^n
  \tilde{Z}^n|\tilde{X}^n\tilde{U}} \| W^n | \bPP{\tilde{X}^n
  \tilde{U}})$.
\end{remark}

\paragraph*{Proof of Theorem~\ref{t:wiretap_subadditivity}}
For any distribution $\bPP{\tU\tX^n\tY^n\tZ^n}$, note first that (see
\cite[Lemma 17.12]{CsiKor11})
\begin{align}
I(\tU \wedge \tY^n) - I(\tU\wedge \tZ^n) = n[I(\tU \wedge \tY_J |
  V_J,J) - I(\tU \wedge \tZ_J | V_J,J)],
\label{e:MI_add}
\end{align}
where $J$ is distributed uniformly over $\{1, ..., n\}$ and $V_j =
(\tY_{j}^-, \tZ_j^+)$. Next, consider
\begin{align}
D(\bPP{\tY^n\tZ^n|\tX^n\tU}\| W^n| \bPP{\tX^n\tU}) &=
D(\bPP{\tY^n|\tX^n\tU}\| W_1^n| \bPP{\tX^n\tU})+
D(\bPP{\tZ^n|\tY^n\tX^n\tU}\| W_2^n| \bPP{\tY^n\tX^n\tU}).
\label{e:divergence_chain}
\end{align} 
The first term on the right is bounded below by
\begin{align*}
&D(\bPP{\tY^n|\tX^n\tU}\| W_1^n| \bPP{\tX^n\tU})-
  D(\bPP{\tZ^n|\tX^n\tU}\| W_2^n| \bPP{\tX^n\tU}) \\ &=
  \mathbb{E}\bigg[ \log
    \frac{W_2^n(\tZ^n|\tX^n)}{W_1^n(\tY^n|\tX^n)}\bigg] +
  H(\tZ^n|\tX^n, \tU) - H(\tY^n|\tX^n, \tU) \\ &= \sum_{j=1}^n \bigg[
    \mathbb{E}\bigg[ \log \frac{W_2(\tZ_j|\tX_j)}{W_1(\tY_j|\tX_j)}
      \bigg] + H(\tZ_j|\tX^n, \tU, V_j) - H(\tY_j|\tX^n, \tU,
    V_j)\bigg] \\ &= \sum_{j=1}^n \bigg[D(\bPP{\tY_j|\tX^n \tU V_j}\|
    W_1|\bPP{\tX^n\tU V_j}) - D(\bPP{\tZ_j|\tX^n \tU V_j}\| W_2 |
    \bPP{\tX^n\tU V_j}) \bigg],
\end{align*}
where the second equality follows from~\eqref{e:CKM_inequality}.  For
the second term on the right-side of \eqref{e:divergence_chain}, we
have
\begin{align*}
D(\bPP{\tZ^n|\tY^n\tX^n\tU} \| W_2^n| \bPP{\tY^n\tX^n\tU}) &
=\sum_{j=1}^n D(\bPP{\tZ_j|\tZ_j^+\tY^n\tX^n\tU}\|
W_2|\bPP{\tZ_j^+\tY^n\tX^n\tU}) \\ &\geq \sum_{j=1}^n
D(\bPP{\tZ_j|\tX^n \tU V_j}\| W_2|\bPP{\tX^n\tU V_j}),
\end{align*}
where the inequality uses the convexity of $D(\bPP{} \|\bQQ{})$ in
$(\bPP{},\bQQ{})$.  Using these bounds with
\eqref{e:divergence_chain}, it follows that
\begin{align}
D(\bPP{\tY^n\tZ^n|\tX^n\tU}\| W^n| \bPP{\tX^n\tU}) &\geq \sum_{j=1}^n
D(\bPP{\tY_j|\tX^n \tU V_j}\| W_1|\bPP{\tX^n\tU V_j}) \nonumber
\\ &\geq \sum_{j=1}^n D(\bPP{\tY_j|\tX_j\tU V_j}\| W_1|\bPP{\tX_j\tU
  V_j}) \nonumber \\ &= n D(\bPP{\tY_J|\tX_J \tU V_J
  J}\|W_1|\bPP{\tX_J\tU V_J J}).
\label{e:divergence_subadd1}
\end{align}
Also,
\begin{align}
D(\bPP{\tY^n\tZ^n|\tX^n\tU}\| W^n| \bPP{\tX^n\tU}) &\geq
D(\bPP{\tZ^n|\tY^n\tX^n\tU}\| W_2^n| \bPP{\tY^n\tX^n\tU}) \nonumber
\\ &= \sum_{j=1}^nD(\bPP{\tZ_j|\tZ_j^+\tY^n\tX^n\tU}\| W_2|
\bPP{\tZ_j^+\tY^n\tX^n\tU}) \nonumber \\ &\geq
\sum_{j=1}^nD(\bPP{\tZ_j|\tY_j\tX_j\tU V_j}\| W_2| \bPP{\tY_j\tX_j\tU
  V_j}) \nonumber \\ &= n D(\bPP{\tZ_J|\tY_J\tX_J\tU V_J J}\| W_2|
\bPP{\tY_J\tX_J\tU V_J J}).
\label{e:divergence_subadd2}
\end{align}
The bounds \eqref{e:divergence_subadd1} and
\eqref{e:divergence_subadd2} yield
\[
2D(\bPP{\tY^n\tZ^n|\tX^n\tU}\| W^n| \bPP{\tX^n\tU}) \geq n
D(\bPP{\tY_J\tZ_J|\tX_J\tU V_J J}\| W| \bPP{\tX_J\tU V_J J}).
\]
Consequently, we have
\begin{align*}
\lefteqn{ I(\tU \wedge \tY^n) - I(\tU\wedge \tZ^n) - 2 \alpha
  D(\bPP{\tY^n\tZ^n|\tX^n\tU}\| W^n| \bPP{\tX^n\tU}) } \\ &\le n \big[
  I(\tU \wedge \tY_J | V_J,J) - I(\tU \wedge \tZ_J | V_J,J) - \alpha
  D(\bPP{\tY_J\tZ_J|\tX_J\tU V_J J}\| W| \bPP{\tX_J\tU V_J J}) \big]
\\ &\le n C_\mathtt{s}^\alpha(W),
\end{align*}
where, in the last inequality, we removed $(V_J,J)$ by taking the
maximum over realizations of $(V_J,J)$.  Since
$\bPP{\tU\tX^n\tY^n\tZ^n}$ is arbitrary, the proof is completed.\qed

\section{Discussion}
Our proofs of strong converse have followed a common recipe where an
important step is to establish the super-additivity and
sub-additivity, respectively, of the lower and upper bounds involving
the changed measure. To facilitate this, we have used appropriately
crafted variational formulae for these bounds which allowed us to
establish the desired additivity properties. These results, Theorem
\ref{theorem:WZ}, Theorem \ref{theorem:function_computation}, Theorem
\ref{thm:CR}, Theorem \ref{thm:SK}, and Theorem
\ref{t:wiretap_subadditivity}, along with
Proposition~\ref{proposition:almost-subadditivity-entropy} may be of
independent interest.

We restricted our treatment to the case of random variables taking
finitely many values. But this assumption was used only to establish
the variational formulae \eqref{eq:WZ-variational-form},
\eqref{e:varitional_function}, \eqref{eq:CR-variational-formula}, and
\eqref{e:wiretap_variational}, and our results will hold whenever
these formulae can be established. In particular, a technical
difficulty in generalizing these formulae is to replace the use of
uniform continuity of the information quantities in our proofs with
suitable conditions. \textchange{We only need this in the neighborhood
  of product distributions; we have not pursued generalization in this
  direction in
  the current paper, but techniques we develop can be applied to more
  general distributions. Regarding continuous channels,
  the strong converse theorems were proved in \cite{FonTan16,FonTan17b,FonTan19} for the Gaussian multiple access channel, the Gaussian broadcast channel,
  and some class of Gaussian networks.}

\textchange{An intriguing direction of research is if the change of
  measure argument can be used to derive second-order converse
  bounds, namely extending the results of~\cite{ Hay09,PolPooVer10} to
  the multiterminal setting.  For
  centralized coding problems, such as the Gray-Wyner network, an
  application of the argument in \cite{GuEff09} (namely, the change of
  measure argument without introducing penalty terms) to each type
  class leads to the exact second-order converse bound
  (cf.~\cite{Wat17b, ZhoTanMot17}).  A difficulty for distributed
  coding problems is the evaluation of the variational formulae; to derive
  second-order bounds, we need to take the limit of block length $n$
  and the multiplier $\alpha$ simultaneously.  Recently, following-up
  on an early version of this paper, an evaluation
  method for the variational formulae was developed in \cite{Ooh19} which
  used the bound in \eqref{eq:final-bound-WZ} to derive a
  second-order converse bound for the Wyner-Ziv problem. } 

The strong converse claim considered in this paper is that the
capacity remains unchanged even if a constant error $0 < \varepsilon <
1$ is allowed. A stronger notion of strong converse, termed the
exponential strong converse or Arimoto converse, requires that the
error converges to $1$ exponentially rapidly when the rate exceeds the
capacity \cite{Ari73}.  In fact, our proofs give exponential strong
converses as well. For instance, in the lossy source coding with
side-information, by setting $\varepsilon=1-2^{-\xi n}$ in the final
bound \eqref{eq:final-bound-WZ} of the proof of Corollary
\ref{corollary:WZ}, we can show that any code with excess distortion
probability less than $1-2^{-\xi n}$ must satisfy
\begin{align} \label{eq:WZ-exponential-converse}
R + \mu D \ge R_{\mathtt{WZ}}^{\mu,\alpha}(\bPP{XY}) - (\alpha+1) \xi.
\end{align}
Suppose that $R + \mu D \le R_{\mathtt{WZ}}^{\mu}(\bPP{XY}) - 2\nu$ for some $\nu > 0$. 
The variational formula \eqref{eq:WZ-variational-form} implies that 
there exists sufficiently large $\alpha$ such that $R_{\mathtt{WZ}}^{\mu}(\bPP{XY}) \le R_{\mathtt{WZ}}^{\mu,\alpha}(\bPP{XY}) + \nu$.
Thus, if we take $\xi$ so that $\xi < \frac{\nu}{\alpha+1}$, then \eqref{eq:WZ-exponential-converse} is violated, which implies
that the excess distortion probability must be larger than $1-2^{-\xi n}$.
However, the above argument does not give an explicit lower bound for the
exponent of the convergence speed. Such an explicit bound has been
derived recently by Oohama for certain multiterminal problems
(cf.~\cite{Ooh15, Ooh16}).


\textchange{Another interesting problem, which we have not considered,
is that of the multiple access channel (MAC). The strong converse for
MAC was established in~\cite{Dueck81, Ahl82} using a technical tool called the
``wringing lemma.'' While we can recast the proof of~\cite{Dueck81, Ahl82} in
our change of measure language, but it does not offer any extra
insight. In particular, we cannot circumvent the wringing lemma and
simplify the proof of~\cite{Dueck81,Ahl82}; indeed, it is of interest to
simplify this opaque and technical proof.
}

\textchange{Finally, it is of interest to examine the applicability of our strong converse 
proof recipe to problems studied in other fields. One such instance was recently
demonstrated in \cite{TyaWat19} where this recipe was used to provide
an alternative proof for the multiprover 
nonsignaling parallel repetition theorem, an important result in theoretical computer 
science and physics.}

Even though we have illustrated the utility of our recipe only for
several representative problems, we believe that this recipe provides
strong converse theorems for any problems as long as single-letter
characterizations of the optimal rates under weak converse are
known. An interesting future direction will be an application of this
change-of-measure argument to problems such that single-letter
characterizations of weak converse are unknown. A partial attempt for
this problem was made in \cite{GuEff11} for centralized coding
problems.\footnote{Instances of strong converses for problems with
  unknown single-letter characterization of the optimal rate are
  available; see \cite{AhlCsi86, KosKli17}. Both these results apply
  the blowing-up lemma in a non-trivial manner.}  A research in such a
direction will establish a folklore in information theory: Strong
converse holds for any stationary memoryless system.

\appendix

\subsection{Proof of variational formaula \eqref{eq:WZ-variational-form}} \label{app:WZ}

The proof is almost the same as the proof of
Proposition~\ref{p:remove_Markov}.  Clearly, the left-side is greater 
than or equal to the right-side. To prove the other direction, for each $\alpha>0$, let
$\mathrm{P}^\alpha_{\tilde{U}\tilde{X}\tilde{Y}\tilde{Z}}$ be the 
minimizer for the inner minimum on the right-side, and let
$\mathrm{Q}^\alpha_{\tilde{U}XY\tilde{Z}} =
\mathrm{P}^\alpha_{\tilde{Z}|\tilde{U}\tilde{Y}}
\mathrm{P}^\alpha_{\tilde{U}|\tilde{X}} \bPP{XY}$ be the induced
distribution.  Since $G(\bPP{\tilde{U}\tilde{X}\tilde{Y}\tilde{Z}}) =
I(\tilde{U} \wedge \tilde{X}|\tilde{Y}) +
\mathbb{E}[d(\tilde{X},\tilde{Z})]$ is nonnegative and bounded above
by $a= \log |\cX| + D_{\max}$, it must hold that
$D(\mathrm{P}^\alpha_{\tilde{U}\tilde{X}\tilde{Y}\tilde{Z}} \|
\mathrm{Q}^\alpha_{\tilde{U}XY\tilde{Z}}) \le (a/\alpha)$.

Furthermore, since $G(\bPP{\tilde{U}\tilde{X}\tilde{Y}\tilde{Z}})$ is
uniformly continuous, there exists a function $\Delta(t)$ satisfying
$\Delta(t) \to 0$ as $t \to 0$ such that
\begin{align*}
R_{\mathtt{WZ}}^{\mu,\alpha}(\bPP{XY}) &\ge
G(\mathrm{P}^\alpha_{\tilde{U}\tilde{X}\tilde{Y}\tilde{Z}}) \\ &\ge
G(\mathrm{Q}^\alpha_{\tilde{U}XY\tilde{Z}}) - \Delta(a/\alpha) \\ &\ge
R_{\mathtt{WZ}}^\mu(\bPP{XY}) - \Delta(a/\alpha).
\end{align*}
Thus, we obtain the desired inequality by taking $\alpha \to \infty$,
which completes the proof. \qed

\subsection{Proof of variational formula \eqref{e:varitional_function}} \label{app:function-computation}

The proof mimics the one above, but has been included for
completeness. As before, it is easy to see that the left-side is greater than or
equal to the right-side. For the other direction, for each $\alpha>0$, let
$\mathrm{P}^\alpha_{\tilde{U}\tilde{V} \tilde{X}\tilde{Y}}$ 
be the minimizer for the inner minimum on the right-side, and let  
$\mathrm{Q}^\alpha_{\tilde{U}\tilde{V}XY} =
\mathrm{P}^\alpha_{\tilde{V}|\tilde{U}\tilde{Y}}
\mathrm{P}^\alpha_{\tilde{U}|\tilde{X}} \bPP{XY}$  
be the induced distribution. Since $G(\bPP{\tU\tV\tX\tY}) = I(\tU,\tV
\wedge \tX|\tY) + I(\tU,\tV \wedge \tY|\tX)$ 
is nonnegative and bounded above by $a= \log|\cX||\cY|$, it must hold 
that 
\begin{align} \label{eq:WZ-divergence-upper-bound}
D(\mathrm{P}^\alpha_{\tilde{U}\tilde{V} \tilde{X}\tilde{Y}} \| \mathrm{Q}^\alpha_{\tilde{U}\tilde{V}XY}) \le \frac{a}{\alpha}
\end{align}
and\footnote{We have
put the subscripts in \eqref{eq:almost-reproduction-conditions} to emphasize the underlying measure.}
\begin{align} \label{eq:almost-reproduction-conditions}
  H_{\mathrm{P}^\alpha}(\tilde{F}|\tY,\tU) \le \frac{a}{\alpha},
\qquad
  H_{\mathrm{P}^\alpha}(\tilde{F}|\tX,\tU,\tV) \le \frac{a}{\alpha}.
\end{align}

Using the compactness of the finite dimensional probability simplex,
there exists a subsequence
$\{\mathrm{Q}^{\alpha_i}_{\tilde{U}\tilde{V}XY} \}_{i=1}^\infty$  
of $\{ \mathrm{Q}^{\alpha}_{\tilde{U}\tilde{V}XY} \}_{\alpha \in \mathbb{N}}$ that converges to
$\mathrm{Q}^*_{\tilde{U}\tilde{V}XY}$. By uniform continuity of the
entropy, \eqref{eq:WZ-divergence-upper-bound} and
\eqref{eq:almost-reproduction-conditions}  
imply that the limit point $\mathrm{Q}^*_{\tilde{U}\tilde{V}XY}$
satisfies $H_{\mathrm{Q}^*}(F|Y,\tU)=H_{\mathrm{Q}^*}(F|X,\tU,\tV) =
0$. Furthermore, since  
$G(\bPP{\tU\tV\tX\tY})$ is also uniform continuous, there exists a
function $\Delta(t)$ satisfying $\Delta(t) \to 0$ as $t \to 0$ such
that 
\begin{align*}
R_f^{\alpha_i}(\bPP{XY}) &\ge G(\mathrm{P}^{\alpha_i}_{\tU\tV\tX\tY}) \\
&\ge G(\mathrm{Q}^{\alpha_i}_{\tU\tV XY}) - \Delta(a / \alpha_i).
\end{align*} 
Thus, by taking the limit $i\to\infty$, we have
\begin{align*}
\sup_{\alpha > 0} R_f^{\alpha}(\bPP{XY}) 
&\ge G(\mathrm{Q}^*_{\tU\tV XY}) \\
&\ge R_f(\bPP{XY}),
\end{align*}
which completes the proof. \qed

\subsection{Proof of variational formula \eqref{eq:CR-variational-formula}}

The proof is a minor variant of those of \eqref{eq:WZ-variational-form} and \eqref{e:varitional_function}. 
We only need to observe that the function $G(\bPP{\tU\tV\tX\tY}) =
I(\tU,\tV \wedge \tX,\tY) - \mu\big( I(\tU,\tV \wedge \tX|\tY) +
I(\tU,\tV \wedge \tY|\tX) \big)$ is bounded above by $a = \log
|\cX||\cY|$ and below by $b = -\mu \log |\cX||\cY|$. With this
observation, the same arguments go through. \qed 

\bibliography{IEEEabrv,references} \bibliographystyle{IEEEtranS}

\end{document}